\documentclass[sigconf]{acmart}

\AtBeginDocument{%
  \providecommand\BibTeX{{%
    \normalfont B\kern-0.5em{\scshape i\kern-0.25em b}\kern-0.8em\TeX}}}

\copyrightyear{2023}
\acmYear{2023}
\setcopyright{acmlicensed}
\acmConference[CIKM '23] {Proceedings of the 32nd ACM International Conference on Information and Knowledge Management}{October 21--25, 2023}{Birmingham, United Kingdom.}
\acmBooktitle{Proceedings of the 32nd ACM International Conference on Information and Knowledge Management (CIKM '23), October 21--25, 2023, Birmingham, United Kingdom}
\acmPrice{15.00}
\acmISBN{979-8-4007-0124-5/23/10}
\acmDOI{10.1145/3583780.3614834}

\usepackage{algorithm}
\usepackage{algorithmicx}
\usepackage{multirow}
\usepackage[normalem]{ulem}
\useunder{\uline}{\ul}{}
\usepackage{xcolor} 
\usepackage{enumitem}
\usepackage{subfigure}
\usepackage{bbding}
\usepackage{amsthm}

\newcommand{\nosection}[1]{\vspace{2pt}\noindent\textbf{#1.}}

\newcommand{\modelname}{DGREC}

\newcommand*{\dif}{\mathop{}\!\mathrm{d}}
\newtheorem{Proposition}{Proposition}
\newtheorem{Theorem}{Theorem}
\usepackage[noend]{algpseudocode}
\usepackage{amsfonts}
\usepackage{amsmath}

\begin{document}

\title{Decentralized Graph Neural Network for Privacy-Preserving Recommendation}


\author{Xiaolin Zheng}
\affiliation{%
  \institution{Zhejiang University}
  \department{Department of Computer Science}
  \city{Hangzhou}
  \state{Zhejiang}
  \country{China}}
\orcid{0000-0001-5483-0366}
\email{xlzheng@zju.edu.cn}

\author{Zhongyu Wang}
\affiliation{%
  \institution{Zhejiang University}
  \department{Department of Computer Science}
  \city{Hangzhou}
  \state{Zhejiang}
  \country{China}}
\orcid{0009-0002-7909-7003}
\email{iwzy7071@zju.edu.cn}

\author{Chaochao Chen}
\authornote{Corresponding author.}
\affiliation{%
  \institution{Zhejiang University}
  \department{Department of Computer Science}
  \city{Hangzhou}
  \state{Zhejiang}
  \country{China}}
\orcid{0000-0003-1419-964X}
\email{zjuccc@zju.edu.cn}

\author{Jiashu Qian}
\affiliation{%
  \institution{Zhejiang University}
  \department{Department of Computer Science}
  \city{Hangzhou}
  \state{Zhejiang}
  \country{China}}
\orcid{0009-0006-4141-3321}
\email{iqjs0124@gmail.com}

\author{Yao Yang}
\affiliation{%
  \institution{Zhejiang Lab}
  \city{Hangzhou}
  \state{Zhejiang}
  \country{China}}
\orcid{0000-0002-7007-9071}
\email{yangyao@zhejianglab.com}

\renewcommand{\shortauthors}{Xiaolin Zheng, Zhongyu Wang, Chaochao Chen, Jiashu Qian, \& Yao Yang}

\begin{abstract}
Building a graph neural network (GNN)-based recommender system without violating user privacy proves challenging.
Existing methods can be divided into federated GNNs and decentralized GNNs. 
But both methods have undesirable effects, i.e., low communication efficiency and privacy leakage.
This paper proposes \modelname, a novel decentralized GNN for privacy-preserving recommendations, where users can choose to publicize their interactions.
It includes three stages, i.e., graph construction, local gradient calculation, and global gradient passing.
The first stage builds a local inner-item hypergraph for each user and a global inter-user graph.
The second stage models user preference and calculates gradients on each local device.
The third stage designs a local differential privacy mechanism named secure gradient-sharing, which proves strong privacy-preserving of users' private data. 
We conduct extensive experiments on three public datasets to validate the consistent superiority of our framework.
\end{abstract}

\ccsdesc[500]{Information systems~Personalization}
\ccsdesc[500]{Security and privacy~Web application security}

\keywords{Recommender system, decentralized graph neural network, privacy protection}


\maketitle

\section{Introduction}
\label{section_introduction}
Personalized recommendations have been applied to many online services, such as advertising, social media, and E-commerce, to solve the information overload problem.
Among them, graph neural network (GNN)-based recommendation methods achieve state-of-the-art performance~\cite{NGCF, LightGCN, GCMC}.
However, all these personalized recommendation methods require users' private data, e.g., interactions and ratings.
Under strict data protection regulations, such as GDPR\footnote{https://gdpr-info.eu}, online platforms is prohibited from collecting private data and hence cannot make accurate recommendations.
To achieve superior recommendation performance and strong privacy protection, building a GNN-based privacy-preserving recommender system is worth investigating.

Existing work can be classified into two types, i.e., federated GNNs~\cite{FedGNN, FeSoG, LPGNN, GWGNN} and decentralized GNNs~\cite{p2pGNN,spread_gnn, DLGNN,decentralized_federated_GNN,central_decentral_system}.
We summarize their main differences from the following three aspects: training mechanism, communication efficiency, and privacy protection.
\textbf{For training mechanism}, federated GNNs need a central server to update the model and synchronize it with all clients.
In contrast, decentralized GNNs require clients to update their models, where each client exchanges parameters with neighbors.
\textbf{For communication efficiency}, federated GNNs have high communication costs in the central server, causing a bottleneck in training speed.
In contrast, decentralized GNNs asynchronously update the models, increasing training efficiency.
\textbf{For privacy protection}, federated GNNs leverage privacy-preserving methods, which guarantee privacy-preserving strength theoretically.
In contrast, decentralized GNNs have privacy leakage problems. 
%

%

The above differences between the two methods motivate us to study a decentralized GNN for privacy-preserving recommendations.
Designing such a GNN is a non-trivial task, given the following three challenges.
\textbf{CH1: }\textit{meeting diverse demands for privacy protection}.
Existing work~\cite{FedGNN, FedREC, FeSoG} takes all user interactions as private, causing setbacks in user experience.
Actually, users are willing to publicize part of their interactions.
For example, on social media, users like to share their moments to gather friends with similar preferences.
\textbf{CH2: }\textit{protecting user privacy cost-effectively}.
Existing privacy-preserving methods, e.g., secret sharing~\cite{secret_sharing} and homomorphic encryption~\cite{homomorphic_encryption}, equip heavy mathematical schemes, causing high communication costs.
\textbf{CH3: }\textit{mitigating the trade-off between effectiveness and efficiency}.
In a decentralized GNN, involving one user for training cannot provide sufficient data to train an effective model, whereas involving too many users slows down the training process.
%

%
This paper proposes \modelname, a decentralized GNN for privacy-preserving recommendations to solve the above challenges.
To overcome \textbf{CH1}, we propose allowing users to publicize their interactions freely.
%
%
\modelname~mainly has three stages, i.e., \textbf{Stage1:} graph construction, \textbf{Stage2:} local gradient calculation, and \textbf{Stage3:} global gradient passing.
In \textbf{Stage1}, all users constitute an \textit{inter-user graph} collaboratively, and each user constructs an \textit{inner-item hypergraph} individually.
For the inter-user graph, edges come from their friendship or proximity.
For the inner-item hypergraph, nodes are the items that the user interacts with or publicizes, and edges come from item tags.
In \textbf{Stage2}, we model user preference based on the inner-item hypergraph to protect user privacy and enhance preference representation.
Specifically, we condense the item hypergraph into an interest graph by mapping items into different interests.
Given that the distilled interests can be noisy, we propose an interest attention mechanism that leverages GNN architecture to minimize noisy interests.
Finally, we pool the interests to model user preference and calculate local gradients.
In \textbf{Stage3}, each user samples multi-hop neighbors in the inter-user graph and trains the models collaboratively to overcome \textbf{CH3}.
We propose a mechanism based on the Local Differential Privacy (LDP) named secure gradient-sharing, which proves strong privacy-preserving of users' private data to overcome \textbf{CH2}.
Specifically, the calculated gradients in Stage 2 are first encoded, then propagated among the neighbourhood, and finally decoded to restore in a noise-free way.
We summarize the main contributions of this paper as follows:
%

\begin{itemize}[leftmargin=*]
\item We propose a novel decentralized GNN for privacy-preserving recommendations.
To the best of our knowledge, this is the first decentralized GNN-based recommender system.
\item We propose secure gradient-sharing, a novel privacy-preserving mechanism to publish model gradients, and theoretically prove that it is noise-free and satisfies Rényi differential privacy.
\item We conduct extensive experiments on three public datasets, and consistent superiority validates the success of the proposed framework.
\end{itemize}
\section{Related work}
\label{Related_work_in_details}
In this section, we discuss prior work on GNN and privacy protection, as we propose a decentralized GNN for privacy-preserving recommendations.
\subsection{GNN} 
Recent studies have extensively investigated the mechanisms and applications of GNN.
For mechanisms, we investigate graph convolutional networks and graph pooling. 
For applications, we investigate centralized GNNs, federated GNNs, and decentralized GNNs.

\nosection{Graph convolutional networks}
Most graph convolutional networks focus on pairwise relationships between the nodes in a graph~\cite{GCN, graphsage}.
However, they neglect the higher-order connectivity patterns beyond the pairwise relationships~\cite{hypergraph, hypergraph_define2}.
\cite{first_introduce_hypergraph} first introduces hypergraph learning as a propagation process to minimize label differences among nodes.
HGNN~\cite{first_hypergcn} proposes a hypergraph neural network by truncated Chebyshev polynomials of the hypergraph Laplacian.
HyperGNN~\cite{HyperGNN} further enhances the capacity of representation learning by leveraging an optional attention module.
Our paper leverages HyperGNN without an attention module to aggregate node features and learn node assignment as we model a user's preferences based on his item hypergraph.

\nosection{Graph Pooling}
Graph pooling is widely adopted to calculate the entire representation of a graph.
Conventional approaches sum or average node embeddings~\cite{gnn_summing_pooling, gnn_topk_pooling, gnn_topk_pooling2}.
However, these methods cannot learn the hierarchical representation of a graph.
To address this issue, DiffPool~\cite{Differential_Pool} first proposes an end-to-end pooling operator, which learns a differentiable soft cluster assignment for nodes at each layer.
MinCutPool~\cite{min_cut_pool} further formulates a continuous relaxation of the normalized Min-cut problem to compute the cluster assignment.
Motivated by DiffPool, our paper designs a different constraint on cluster assignment: one node can be mapped into multiple clusters, but each cluster should be independent.

\nosection{Centralized GNNs}
Centralized GNNs have been vastly applied in recommendation scenarios.
Leveraging a graph structure can enhance representations for users and items~\cite{NGCF} and capture complex user preferences behind their interactions~\cite{Sequence_GNN}.
PinSage~\cite{pinsage} and LightGCN~\cite{LightGCN} refine user and item representations via multi-hop neighbors' information.
%
%
To deal with noisy interactions and node degree bias, SGL\cite{SGL} explores self-supervised learning on a user-item graph and reinforces node representation via self-discrimination.
Although these methods achieve superior performance, they cannot be applied in privacy-preserving recommendation scenarios as most user interactions are private.
In contrast, our paper designs a decentralized GNN for privacy-preserving recommendations.

\nosection{Federated GNNs}
Federated GNNs involve a central server to orchestrate the training process and leverage an LDP mechanism~\cite{DP_current_Fed} to protect user privacy.
LPGNN\cite{LPGNN} and GWGNN\cite{GWGNN} require clients to publish local features to the server, where feature aggregation is executed.
However, these two methods violate the privacy-preserving requirements in recommendation scenarios: local interactions should not be published.
FedGNN~\cite{FedGNN} assumes a trustworthy third party to share user-item interactions securely and adds Gaussian noises to protect user privacy.
However, finding a credible third party in real scenarios is difficult and hence limits the application of this method.
FeSoG~\cite{FeSoG} shares user features to make recommendations and adds Laplacian noises to preserve users' privacy.
However, user features are sensitive data as malicious participants can leverage these data to infer private user interactions~\cite{Prisvr}.
Federated GNNs have high communication costs in the central server, causing training speed bottlenecks.
In contrast, our method adopts a decentralized training mechanism to improve training efficiency, where each client updates their model asynchronously.

\nosection{Decentralized GNNs}
Decentralized GNNs require clients to cooperate in training prediction models.
SpreadGNN~\cite{spread_gnn} and DLGNN~\cite{DLGNN} calculate model gradients with local interactions and share gradients among clients to update models.
D-FedGNN~\cite{decentralized_federated_GNN} introduces Diffie-hellman key exchange to secure communication and utilizes decentralized stochastic gradient descent algorithm~\cite{DPSGD} to train models.
P2PGNN~\cite{p2pGNN} makes local predictions and uses page rank to diffuse predictions.
All these methods exchange gradients without using a privacy-preserving mechanism.
Consequently, malicious neighbors can infer users' private data from received gradients.
In contrast, our method proposes an LDP mechanism to secure shared gradients, protecting user privacy.

\begin{figure*}[t]
\centering
\includegraphics[width=0.85\linewidth]{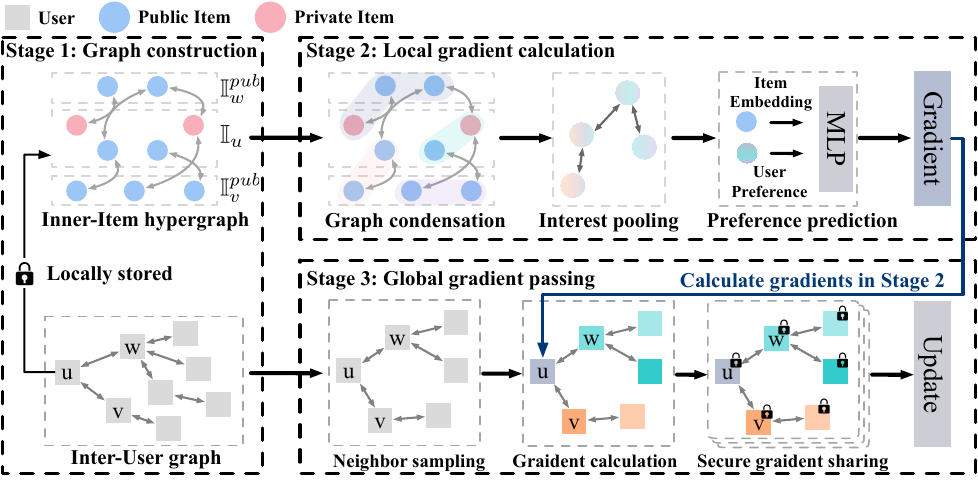}
\caption{
The framework of our proposed \modelname.
Here, we show the training process of user $u$.
After neighbor sampling, all sampled users calculate gradients in Stage 2 simultaneously.
}
\label{figure_framework_overall_architecture}
\end{figure*}

\subsection{Privacy-preserving mechanisms}
Preserving users' privacy is a topic with increasing attention, including the Differential Privacy (DP)~\cite{dp0, dp1, renyi_differential_privacy,chen2022differential,DBLP:conf/aaai/LiaoLZY023}, Homomorphic Encryption~\cite{HE1, HE2, HE3}, and Secure Multi-party Computation~\cite{MPC1, MPC2}.
Among them, Local Differential Privacy, an implementation of DP, has been embraced by numerous machine-learning scenarios~\cite{renyi_differential_privacy}, where noise is added to individual data before it is centralized in a server.
\cite{dp0}~gives a standard definition of $\epsilon$-DP, where $\epsilon$ is the privacy budget. 
\cite{dp1}~further defines $(\epsilon,\delta)$-dp to give a relaxation of $\epsilon$-DP, where $\delta$ is an additive term and $1-\delta$ represents the probability to guarantee privacy.
To the best of our knowledge, most existing privacy-preserving GNN-based recommender systems~\cite{FedGNN, FeSoG} define their privacy protection based on $(\epsilon,\delta)$-DP.
However, $(\epsilon, \delta)$-DP has an explosive privacy budget for iterative algorithms, and hence it cannot guarantee privacy protection for multiple training steps.
To solve this problem, we take $(\alpha, \epsilon)$-RDP~\cite{renyi_differential_privacy} to guarantee tighter bounds on the privacy budget for multiple training steps.
Our work proves that the secure gradient-sharing mechanism satisfies $(1.5, \epsilon)$-RDP, guaranteeing stronger privacy protection in real scenarios.
\section{Method}
\label{section_method}
As seen in Figure~\ref{figure_framework_overall_architecture}, the framework of \modelname~consists of three stages, i.e., graph construction, local gradient calculation, and global gradient passing.
The first stage constructs an inter-user graph and an inner-item hypergraph.
The second stage models user preference and calculates local gradients. 
The third stage proposes a novel privacy-preserving mechanism named secure gradient-sharing.
Finally, users update their local models with the decoded gradients.
\modelname~achieves competitive recommendation performance and meets diverse demands for privacy protection.

\subsection{Preliminary}
\nosection{Rényi differential privacy (RDP)}
A randomized algorithm $M$ satisfies $(\alpha, \epsilon)$-Rényi differential privacy~\cite{renyi_differential_privacy}, if for any adjacent vectors $x$ and $x^{\prime}$, it holds that
$
D_{\alpha}(M(x)\parallel M(x^{\prime})) \le \epsilon$, with $\epsilon$ denoting the privacy budget and $\alpha$ denoting the order of Rényi divergence.
Small values of $\epsilon$ guarantee higher levels of privacy, whereas large values of $\epsilon$ guarantee lower levels of privacy.
We can define RDP for any $\alpha \geq 1$.

\nosection{Problem formulation}
Assume we have a set of users and items, denoted by $\mathbb{U}$ and $\mathbb{I}$.
Generally, user $u$ has a set of neighbor users $\mathcal{N}_{u}\mathrm{:}\{u_{1}, u_{2}, \cdots, u_{m}\}$ and interacted items $\mathbb{I}_{u}\mathrm{:}\{i_{1}, i_{2}, \cdots, i_{n}\}$,  with $m$ and $n$ denoting the number of users and items.
User $u$ can decide whether to publicize their interactions, and we let $\mathbb{I}^{pub}_{u}$ denote his publicized interactions.
The decentralized privacy-preserving recommendation aims to recommend items matching user preferences by leveraging the user's interactions and his neighbors' publicized interactions. 


\subsection{Graph construction}
\label{section_building_graph}
In the first stage, we construct the inter-user graph and inner-item hypergraph. 
The inter-user graph will support neighbor sampling and secure gradient-sharing in the third stage.
Each user trains his recommendation model collaboratively with his multi-hop neighbors, who are more likely to share similar preferences.
The inner-item hypergraph will support modeling user preference in the second stage.

\nosection{Inter-user graph construction}
Recent studies demonstrate that user friendship or proximity reveals behavior similarity between users~\cite{DMF, FeSoG}.
Driven by this motivation, we construct a global inter-user graph, utilizing the communication protocol as described in reference~\cite{p2pGNN}. 
As delineated in "Section B: Communication Protocol" of the reference, the implemented protocol demonstrates proficiency in handling both static and dynamic situations.
Each user can anonymously disseminate messages among multi-hop neighbors by communicating with his 1-hop neighbors.
%
The inter-user graph is defined as $\mathcal{G}=\{\mathcal{V}, \mathcal{E}\}$, where $\mathcal{V}$ denotes the user set and $\mathcal{E}$ denotes the edge set. 
We set $\mathcal{N}_{u} = \{v:(u,v) \in \mathcal{E}\}$ to denote the neighbors of user $u$.

\nosection{Inner-item hypergraph construction}
To make privacy-preserving recommendations, we model user preferences based on limited interactions, as users' features and most interactions are private.
By representing user interactions as a graph, it is easier to distinguish his core and temporary interests, as core interests result in frequent and similar interactions.
Here, we build an item hypergraph for each user, where a hyperedge connects multiple items~\cite{hypergraph}.
To stabilize model training, we further leverage tag information to establish item relationships.

The inner-item hypergraph for user $u$ is defined as $\mathbb{G}_{u}=\{\mathbb{V}_{u},\mathbb{A}_{u}\}$, where $\mathbb{V}_{u}=\{i \mid i \in \mathbb{I}_{u} \cup \mathbb{I}^{pub}_{v},\,v \in \mathcal{N}_{u}\}$ denotes the item set and $\mathbb{A}_{u}$ is the incidence matrix.
Each entry $\mathbb{A}_{u}(i, t)$ indicates whether item $i \in \mathbb{V}_{u}$ is connected by a hyperedge $t \in \mathbb{T}$.
In our setting, $\mathbb{T}$ corresponds to a tag set.
We set $\mathbb{A}_{u}(i,t) = 1$ if item $i$ has tag $t$ and $\mathbb{A}_{u}(i,t) = 0$ otherwise.
We let $deg(i)$ denote the degree of item $i$ and $deg(t)$ denote the degree of hyperedge $t$, where $deg(i)=\sum_{t \in \mathbb{T}}\mathbb{A}_{u}(i, t)$ and $deg(t)=\sum_{i \in \mathbb{V}_{u}}\mathbb{A}_{u}(i, t)$.
We let $D_{v} \in \mathbb{R}^{|\mathbb{V}_{u}|\times|\mathbb{V}_{u}|}$ and $D_{t} \in \mathbb{R}^{|\mathbb{T}|\times|\mathbb{T}|}$ denote diagonal matrices of the item and hyperedge degrees, respectively.

\subsection{Local gradient calculation}
\label{section_local_calculation}
In the second stage, we model each user's preference for an item by his inner-item hypergraph and calculate gradients based on ground-truth interactions.

\nosection{Modelling user preference}
We model user preference through three steps, i.e., graph condensation, interest pooling, and preference prediction.

\nosection{(1) Graph condensation}
One user interaction provides weak signals for his preferences.
To model user preferences, an intuitive way is to gather these weak signals into strong ones.  
Motivated by DiffPool\cite{Differential_Pool}, we condense his item hypergraph into an interest graph by learning a soft cluster assignment.

We first aggregate features among items and learn an item assignment by hypergraph convolutional networks.
We define the aggregation and assignation process as,
\begin{equation}
\begin{aligned}
E^{\prime}_{u} &= D^{-1/2}_{v}\mathbb{A}_{u}D^{-1}_{t}\mathbb{A}_{u}^{T}D^{-1/2}_{v}E_{u}W_{1},\\ 
S_{u}&=\mathrm{softmax}(D^{-1/2}_{v}\mathbb{A}_{u}D^{-1}_{t}\mathbb{A}_{u}^{T}D^{-1/2}_{v}E_{u}W_{2}),\nonumber
\end{aligned}
\end{equation}
%
%
where $E^{\prime}_{u} \in \mathbb{R}^{|V_{u}|\times d_{i}}$ denotes the aggregated features,
$S_{u} \in \mathbb{R}^{|V_{u}|\times n_{i}}$ determines the probabilities of mapping items to different interests,
and $E_{u} =\{e_{i} \mid i \in \mathbb{V}_{u}\}$ denotes item embeddings with $e_{i} \in R^{d}$.
$d_{i}$ denotes interest dimension, $n_{i}$ denotes interest number, and $d$ denotes embedding dimension.
$W_{1} \in \mathbb{R}^{d \times d_{i}}$ and  $W_{2} \in \mathbb{R}^{d \times n_{i}}$ are trainable transformation matrices.
%

%
We then condense the item hypergraph into an interest graph and define the condensation process as, 
\begin{equation}
\begin{aligned}
\vec{E}_{u} &= S_{u}^{T}\cdot{E}^{\prime}_{u},\\ \vec{\mathbb{A}}_{u} &= S_{u}^{T}\cdot\mathbb{A}_{u}\cdot S_{u}, \nonumber
\end{aligned}
\end{equation}
where $\vec{E}_{u} \in \mathbb{R}^{|n_{i}|\times d_{i}}$ denotes the interest embedding.
$\vec{\mathbb{A}}_{u} \in \mathbb{R}^{|n_{i}| \times |n_{i}|}$ denotes the adjacency matrix of the interest graph, with each entry $\vec{\mathbb{A}}_{u}(i,j)$ denoting the connectivity strength between interest $i$ and interest $j$.
We continue to let $\mathbb{N}_{i}$ denote the neighbors of interest $i$ to simplify the description.

Different from DiffPool~\cite{Differential_Pool}, our approach assigns an interacted item to multiple interests, as one interaction may be driven by different interests.
However, one problem remains unsolved: the interests may be redundant, i.e., an interest can be represented by a combination of other interests, decreasing convergence speed.
To solve this problem, we leverage Pearson product-moment correlation (Pearson loss)~\cite{correlation}, encouraging the interests to be independent.
We define Pearson loss $\mathcal{L}^{d}_{u}$ as,
\begin{equation}
    \mathcal{L}^{d}_{u} = \frac{1}{n^{2}_{i}}\sum^{n_{i}}_{j=0}\sum^{n_{i}}_{k=0}cov(\vec{E}_{u,j}, \vec{E}_{u,k}) / \sqrt{var(\vec{E}_{u,j})var(\vec{E}_{u,k})},
\label{equation_pearson_loss}
\end{equation}
where $\vec{E}_{u,j}$ denotes the embedding of interest $j$,
$cov(\cdot,\cdot)$ denotes the distance covariance between two interest embeddings, and $var(\cdot)$ denotes the distance variance.

\nosection{(2) Interest pooling}
Some interests are noisy and hence cannot represent user preference, resulting from two aspects.
(1) Noisy behaviors commonly exist in user interactions, e.g., random clicks~\cite{NGCF}.
(2) Neighbor users do not constantly share the same interests.
Thus, we need to limit the weight of noisy interests before pooling them to represent user preference.

Here, we propose a new attention mechanism, named \textit{interest attention}, under a GNN architecture.
Our motivation is that noisy interests are separate from normal interests.
We can use neighbor interests to measure the weight of an interest.
Although some attention mechanisms~\cite{attention, soft_attention} are commonly adopted to achieve this goal, they neglect the relationships between the interests in the interest graph.
We define the process of interest attention as,
\begin{equation}
\begin{aligned}
e_{u} &= \mathrm{softmax}(\vec{P}_{u}^{T})\vec{E}_{u}, \\ \vec{P}_{u,i} &= \vec{E}_{u,i}W_{3} + \sum_{j\in \mathbb{N}_{i}}\vec{\mathbb{A}}_{u}(i,j)(\vec{E}_{u,i}W_{4}-\vec{E}_{u,j}W_{5}), \nonumber
\end{aligned}
\end{equation}
where $e_{u}$ denotes the user preference and $\vec{P}_{u} \in \mathbb{R}^{n_{i} \times 1}$ denotes interest weights.
$W_{3}$, $W_{4}$, and $W_{5}$ are trainable transformation matrices with dimension $\mathbb{R}^{d_{i}\times 1}$.

\nosection{(3) Preference prediction}
We concatenate user preference and item embedding ($e_{u}$ and $e_{i}$) and feed them into a multi-layer perceptron (MLP) to estimate the probability of user $u$ interacting with item $i$.
We define preference prediction as $\hat{y}_{u, i} = \mathrm{MLP}(e_{u} \parallel e_{i})$, where $\parallel$ denotes the concatenation function. 

\nosection{Gradient calculation}
We take Bayesian personalized ranking (BPR) loss~\cite{BPRLoss} as the loss function to learn model parameters. 
It assumes that the observed interactions should be assigned higher prediction values than unobserved ones.
The BPR loss $\mathcal{L}_{u}^{p}$ is calculated as,
\begin{equation}
    \mathcal{L}^{p}_{u} = \frac{1}{|\mathbb{I}_{u}|}\sum_{i \in \mathbb{I}_{u}, j \in \mathbb{I}^{-}_{u}}-\mathrm{ln}\sigma(\hat{y}_{u,i} - \hat{y}_{u,j}),
\label{equation_training_loss}
\end{equation}
where $\mathbb{I}^{-}_{u}$ denotes the unobserved interactions of user $u$ and $\sigma$ denotes the sigmoid function.

Finally, we calculate local loss $\mathcal{L}_{u}$ for user $u$ as,
\begin{equation}
\mathcal{L}_{u} = \mathcal{L}^{p}_{u} + \mathcal{L}_{u}^{d} + \lambda||\Theta_{u}||_{2},
\label{equation_local_Loss}
\end{equation}
where $\lambda$ is the weight of L2 regularization and $\Theta_{u}$ denotes trainable parameters of the recommendation model for user $u$.
\subsection{Global gradient passing}
\label{subsection_global_gradient_passing}
In the third stage, we propose neighbor sampling and secure gradient-sharing based on the inter-user graph to enhance model performance and training efficiency in a privacy-preserving way.

\nosection{Neighbor sampling}
An intuitive way of decentralized recommendations is that users collaboratively train models with their neighbors.
However, two problems remain unresolved, damaging recommendation performance:
(1) Some item embeddings are not fully trained or even not trained, as interacted items are sparse in recommendation scenarios.
(2) Different users have different convergence speeds.
For example, popular users are more likely to be involved in training.
Inadequately-trained models can bias fully-trained models by shared gradients.
Thus, we design a sampling strategy to upgrade the neighbor-based collaborative training into the neighborhood-based one.

We describe a neighbor sampling strategy for user $u \in \mathcal{V}$ in Algorithm~\ref{algorithm_neighborhood_calculation}.
%
%
%
%
The probability for user $u$ to sample his neighbor $v \in \mathcal{N}_{u}$ is defined as,
\begin{equation}
p_{u,v} = \frac{\mathcal{L}_{v} / (\mathrm{ln}(cnt_{v} + 1) + 1)}{\sum_{w \in \mathcal{N}_{u}} \mathcal{L}_{w} / (\mathrm{ln}(cnt_{w} + 1) + 1)},
\label{equation_sample_probability}
\end{equation}
where $\mathcal{L}_{v}$ is the training loss for user $v$ as calculated in Equation~\eqref{equation_local_Loss}.
We let $cnt_{v}$ denote the training iterations for user $v$.
The sampled probability increases with local training loss and decreases with training iterations.
Thus, underfit models are more likely to be sampled for training.

We employ $U$ to denote all sampled users, i.e., $U=\{v | size(\mathcal{N}_{v}^{s}) > 0\}$.
Each sampled user $v \in U$ simultaneously calculates his local gradients $g_{v}$ based on the loss defined in Equation~\eqref{equation_local_Loss}.

\nosection{Secure gradient-sharing}
Inspired by the one-bit encoder~\cite{motivation_one_bit_encoder}, we propose a novel privacy-preserving mechanism to share gradients in the neighborhood efficiently.
Unlike the one-bit encoder, our mechanism shares the calculated gradients multiple times and satisfies RDP.
The secure gradient-sharing involves three steps, i.e., gradient encoding, gradient propagation, and gradient decoding.

\nosection{(1) Gradient encoding} 
We first encode the calculated gradients to protect user privacy and minimize communication costs.
In the training neighborhood, each user clips his local gradients into $\left[-\delta, \delta\right]$ and takes the Bernoulli distribution to sample the mapped encoding for gradients.
Let $\beta$ denote the perturbation strength to protect user privacy.
We define the gradient encoding for user $u$ as,
\begin{equation}
    g_{u}^{*} \sim 2*\mathrm{Bernoulli}(\frac{1}{e^{\beta} + 1} + \frac{(e^{\beta} - 1)(\mathrm{clip}(g_{u}, \delta) + \delta)}{2(e^{\beta} + 1)\delta}) - 1,
\label{equation_gradient_encoder} 
\end{equation}
where $g_{u}^{*}$ denotes the encoded gradients of user $u$ with each element either equalling to $-1$ or $1$.
Large gradients are more likely to be mapped to $1$, whereas small gradients are more likely to be mapped to $0$.
\begin{Theorem}
For a gradient with a size of $n_{s}$, our secure gradient-sharing mechanism satisfies $(1.5, \epsilon)$-RDP, where $\epsilon = 2n_{s}log(\frac{1.5\pi \delta (e^{\beta} + 1)}{2(e^{\beta} - 1)}\\ + \frac{e^{-0.5\beta} + e^{1.5\beta}}{e^{\beta} + 1})$.
\label{theorem_gradient_encoding_rdp}
\end{Theorem}
\begin{proof}
Let $\mathcal{M}(g)$ be the secure gradient-sharing as described in Equation~\eqref{equation_gradient_encoder}.
We need to show that for any two input gradients $g$ and $g^{\prime}$, we have $D_{\alpha}(M(g)\parallel M(g^{\prime})) \le \epsilon$.
We set $z=e^{\beta}+1$ to simplify the description.
For two continuous distributions defined over the real interval with densities $p$ and $q$, we have
\begin{equation}
\begin{aligned}
    p(g) &= \mathrm{P}[\mathcal{M}(g)=1] = \frac{\delta z + (z-2)g}{2\delta z}, \\
    q(g) &= \mathrm{P}[\mathcal{M}(g)=-1] = \frac{\delta z - (z-2)g}{2\delta z}. \nonumber
\end{aligned}
\end{equation}
Given that we clip gradients $g$ into $\left[-\delta, \delta \right]$, we evaluate $D_{\alpha}(M(g)\parallel M(g^{\prime}))$ separately over three intervals, i.e., $\left(-\infty, -\delta \right)$, $\left[-\delta, \delta \right]$, and $\left(\delta, +\infty \right)$.
Let $n_{s}$ denote the size of the gradient.
We have
\begin{equation}
\begin{aligned}
D_{\alpha}(M(g)\parallel M(g^{\prime})) 
&= \frac{1}{\alpha - 1} \sum \limits_{i=0}^{ns} log\int_{-\infty}^{\infty} p(g_{i})^{\alpha}q(g_{i})^{1-\alpha}\mathrm{d} g_{i}.\nonumber
\end{aligned}
\end{equation}
For intervals $\left(-\infty, -\delta \right)$ and $\left(\delta, +\infty \right)$, we have
\begin{equation}
\begin{aligned}
\int_{-\infty}^{-\delta}p(g)^{\alpha}q(g)^{1-\alpha}\dif g &= \frac{(z-1)^{1-\alpha}}{z},\\
\int_{\delta}^{\infty}p(g)^{\alpha}q(g)^{1-\alpha}\dif g &= \frac{(z-1)^{\alpha}}{z} \nonumber.
\end{aligned}
\end{equation}
For intervals $\left[-\delta, \delta \right]$, we set $x=(z-2)g$ and $y=\delta z$ and have
\begin{small}
\begin{equation}
\begin{aligned}
\int_{-\delta}^{\delta}p(g)^{\alpha}q(g)^{1-\alpha}\dif g &= \frac{1}{2\delta z(z-2)}\int_{-y + 2\delta}^{y - 2\delta}(x+y)^{\alpha}(-x+y)^{1-\alpha} \dif x  \\
&\le \frac{1}{2\delta z(z-2)}\int_{-y}^{y}(x+y)^{\alpha}(-x+y)^{1-\alpha} \dif x \nonumber.
\end{aligned}
\end{equation}
\end{small}
Given the reverse chain rule, we set $t=(x+y)/2y$ and have 
\begin{equation}
\begin{aligned}
\int_{-\delta}^{\delta}p(g)^{\alpha}q(g)^{1-\alpha}\dif g &= \frac{1}{z-2}\int_{0}^{1}(2yt)^{\alpha}(-2yt+2y)^{1-\alpha} \dif t \\
&= \frac{2\delta z}{z-2}\int_{0}^{1}(t)^{\alpha}(-t+1)^{1-\alpha} \dif t \nonumber.
\end{aligned}
\end{equation}
According to the definition of beta and gamma functions, we have
\begin{equation}
\begin{aligned}
\int_{-\delta}^{\delta}p(g)^{\alpha}q(g)^{1-\alpha}\dif g &= \frac{2\delta z}{z-2}\mathcal{B}(\alpha+1, 2-\alpha) \\
&= \frac{2\delta z}{z-2}\frac{\Gamma(\alpha+1)\Gamma(2-\alpha)}{\Gamma(3)} \\
&= \frac{2\delta z\alpha(\alpha-1)}{(z-2)}\frac{\Gamma(\alpha-1)\Gamma(1-(\alpha-1 ))}{\Gamma(3)} \nonumber.
\end{aligned}
\end{equation}
By the reflection formula for the gamma function, we have 
\begin{equation}
\begin{aligned}
\int_{-\delta}^{\delta}p(g)^{\alpha}q(g)^{1-\alpha}\dif g &= \frac{2\delta z\alpha(\alpha-1)\pi}{(z-2)\Gamma(3)sin((\alpha-1) \pi)}\nonumber.
\end{aligned}
\end{equation}
To calculate $D_{\alpha}(M(g)\parallel M(g^{\prime}))$, here we set $\alpha = 1.5$ and have
\begin{small}
\begin{equation}
\begin{aligned}
D_{\alpha}(M(g)\parallel M(g^{\prime})) = 2n_{s}log(\frac{1.5\pi \delta (e^{\beta} + 1)}{2(e^{\beta} - 1)} + \frac{e^{-0.5\beta} + e^{1.5\beta}}{e^{\beta} + 1}) \le \epsilon \nonumber,
\end{aligned}
\end{equation}
\end{small}
which concludes the proof of $(1.5, \epsilon)$-RDP.
\end{proof}

\begin{Theorem}
Given $T$ training iterations, the privacy budget of the secure gradient-sharing mechanism is bounded.
\label{theorem_gradient_encoding_privacy_budget}
\end{Theorem}
\begin{proof}
Given $T$ training iterations, the secure gradient-sharing satisfies $(1.5, \epsilon T)$-RDP according to the sequential composition~\cite{sequential_composition}.
One property of RDP~\cite{renyi_differential_privacy} is that if $\mathcal{M}(g)$ satisfies $(\alpha, \epsilon)$-RDP, for $\forall \gamma > 0$, $\mathcal{M}(g)$ satisfies $(\epsilon^{\prime}, \gamma)$-DP with $\epsilon^{\prime} = \epsilon + \frac{log(1/\gamma)}{\alpha-1}$. 
Thus, the secure gradient-sharing mechanism satisfies $(\epsilon T+2log(1/\gamma ), \gamma)$-DP.
We can further choose large $\gamma$ to reduce the privacy budget, which concludes that the privacy budget is bounded.
\end{proof}
%

\nosection{(2) Gradient propagation} 
We devise a topology-based diffusion method named \textit{gradient propagation} to share the encoded gradients among the neighborhood.
The gossip training protocol has been widely applied in decentralized learning~\cite{gossip, gossip_2}, where users exchange gradients with their $1$-hop neighbors, given its robustness to device breakdown.
However, the encoded gradients are highly inaccurate in our setting, as the number of averaged gradients is insufficient to reduce noise.
%
\begin{algorithm}[t]
\begin{algorithmic}[1]
\Require{
User $u$, number of sampling hops $H$ and sampling users $n_{u}$
} 
\Ensure{
%
The sampled neighbors $\mathcal{N}^{s}_{v}$ for each user $v \in \mathcal{V}$
}
\State Initialize $\mathcal{N}_{v}^{s}$ as empty for $v \in \mathcal{V}$
\State $q\gets \mathrm{queue([u])}$
\For{$h \gets 1$ to $H$ \label{line_start_construct_neighborhood}}
    \State $sz \gets |q|$
    \For{$i \gets 1$ to $sz$}
        \State $v \gets q$.pop()
        \State $ws \gets$ use probabilities calculated in Equation~\eqref{equation_sample_probability} to sample $n_{u}$ users from $\mathcal{N}_{v}$ \label{line_sample}
        \While{$w \in ws$}
            \State add $w$ to $q$ and $\mathcal{N}^{s}_{v}$
            \State add $v$ to $\mathcal{N}^{s}_{w}$
        \EndWhile
    \EndFor
\EndFor \label{line_end_construct_neighborhood}
\State \textbf{return} $\mathcal{N}^{s}_{v}$ for user $v \in U$
\end{algorithmic}
\caption{Neighbor sampling for user $u$}
\label{algorithm_neighborhood_calculation}
\end{algorithm}
\begin{algorithm}[t]
\begin{algorithmic}[1]
\Require{
Encoded gradients $g^{*}_{u}$ for user $u \in U$,
number of sampling hops $H$,
sampling neighbors $\mathcal{N}_{u}^{s}$ for $u \in U$
} 
\Ensure{
Aggregated gradients $\vec{g}_{u}$ for user $u \in U$
}
\For{user $u \in U $ \textbf{parallel}}
\State $\vec{g}_{u} \leftarrow \{g^{*}_{u}\}$
\EndFor
\For{$h \gets 1$ to $2H$}
	\For{user $u \in U $ \textbf{parallel}}
    	\State send $\vec{g}_{u}$ to neighbor $v \in \mathcal{N}^{s}_{u}$ \label{algorithm_gradient_backward_aggregate}
	    \State$\vec{g}_{u} \leftarrow \vec{g}_{u} \cup \vec{g}_{v}, v \in \mathcal{N}^{s}_{u}$
	     \label{algorithm_gradient_backward_send}
	\EndFor
\EndFor
\State \textbf{return} $\vec{g}_{u}$ for user $u \in U$
\end{algorithmic}
\caption{Gradient propagation}
\label{algorithm_gradient_propagating} 
\end{algorithm}

To solve the above issue, we extend the gossip training protocol to disseminate gradients among the neighborhood.
Algorithm~\ref{algorithm_gradient_propagating} describes the process of gradient propagation. 
Similar to the message-passing architecture in GNN~\cite{graphsage, NGCF}, each user in the neighborhood sends and receives gradients in parallel (line \ref{algorithm_gradient_backward_aggregate} and \ref{algorithm_gradient_backward_send}).
Given the longest path in the training neighborhood is $2H$, the encoded gradients are disseminated among all users after $2H$ steps.

\nosection{(3) Gradient decoding}
Conventional DP mechanisms~\cite{FeSoG, FedGNN, FedREC} protect user privacy at the cost of recommendation performance as it introduces noises to model training.
We propose gradient decoding to decode received gradients in a noise-free way.

After gradient propagation, each user locally decodes the received gradients and updates their model with the decoded gradients by stochastic gradient descent (SGD) optimizer.
Here, we describe the decoding process for user $u$ as,
\begin{equation}
\bar{g}_{u} = \delta(e^{\beta} + 1)\mathrm{mean}(\vec{g}_{u}) / (e^{\beta} - 1).
\label{equation_gradient_decoder}
\end{equation}

\begin{Proposition}
The secure gradient-sharing mechanism is unbiased.
\label{proposition_gradient_encoder}
\end{Proposition}
\begin{proof}
To prove the secure gradient-sharing mechanism is noise-free, we need to verify that $E(\bar{g}) = E(g)$.
Since the gradients are encoded from the Bernoulli distribution in Equation~\eqref{equation_gradient_encoder}, we have
\begin{equation}
\begin{aligned}
    E(g^{*}) =\frac{e^\beta-1}{e^\beta+1}\frac{E(g) + \delta}{\delta} -\frac{e^\beta-1}{e^\beta+1}
    = \frac{(e^\beta-1)E(g)}{(e^\beta+1)\delta} \nonumber.
\end{aligned}
\end{equation}
Combining Equation (6), we then have 
\begin{equation}
\begin{aligned}
    E(\bar{g}) &= \frac{\delta(e^\beta + 1)}{e^\beta - 1}E(g^{*}) \\
    &= \frac{\delta(e^\beta + 1)}{e^\beta - 1}\frac{(e^\beta-1)E(g)}{(e^\beta+1)\delta} \\
    &= E(g) \nonumber,
\end{aligned}
\end{equation}
which concludes that the secure gradient-sharing mechanism is noise-free.
\end{proof}

\begin{Proposition}
The secure gradient-sharing mechanism can approach the same convergence speed as centralized SGD.
\label{proposition_convergence}
\end{Proposition}
Proposition~\ref{proposition_convergence} guarantees the convergence of our proposed method, and please refer to paper~\cite{convergence_proof} for its detailed proof.

\section{Analysis on communication cost}
\label{subsection_communication_cost}
Our proposed method has the lowest communication cost among all competitors.
Here, we provide a comparison of each method's communication cost theoretically.
Note that each user samples $n_{u}$ users from his neighbors, the training process involves $H$-hop users, and the gradient to be sent has a size of $n_{s}$.
(1) Our method requires clients to cooperate during model training and leverages a gradient encoder to minimize the communication cost.
Each user has the worst communication cost as $2Hn_{s}(1-n_{u}^H)/(1-n_{u})$. 
(2) Federated methods involve a central server to orchestrate the training process.
They have a communication cost as $d_{r}n_{s}(1-n_{u}^H)/(1-n_{u})$, where $d_{r}$ is the number of bits to express a real number, e.g., $64$.
(3) For decentralized methods, each user has a communication cost as $d_{r}n_{s}(1-n_{u}^H)/(1-n_{u})$.
In summary, our method consistently beats federated and decentralized methods since hop number $H$ is generally a small number, e.g., $3$. 
\section{Experiment}
\label{section_experiment}
\begin{table}[t]
\centering
\caption{Statistics of datasets.}
\begin{tabular}{llll}
\hline
Dataset        & \multicolumn{1}{c}{Flixster} & \multicolumn{1}{c}{\begin{tabular}[c]{@{}c@{}}Book-crossing\end{tabular}} & \multicolumn{1}{c}{Weeplaces} \\ \hline
\# Users        & 3,060                        & 3,060                                                                        & 8,720                         \\ \hline
\# Items        & 3,000                        & 5,240                                                                        & 7,407                         \\ \hline
\# Interactions & 48,369                       & 222,287                                                                       & 546,781                       \\ \hline
Sparsity       & 0.9946                       & 0.9860                                                                       & 0.9915                        \\ \hline
\end{tabular}
\label{exp_dataset_statics}
\end{table}
\begin{table*}[t]
\centering
\caption{
Performance comparison of \modelname~with state-of-the-art methods. 
In one column, the bold values correspond to the methods with the best and runner-up performances.
}
\begin{tabular}{c|ll|ll|ll}
\hline
               & \multicolumn{2}{c|}{Flixster}                          & \multicolumn{2}{c|}{Book-Crossing}                     & \multicolumn{2}{c}{Weeplaces}                          \\
               & \multicolumn{1}{c}{recall} & \multicolumn{1}{c|}{ndcg} & \multicolumn{1}{c}{recall} & \multicolumn{1}{c|}{ndcg} & \multicolumn{1}{c}{recall} & \multicolumn{1}{c}{ndcg} \\ \hline
NeuMF          & 3.28 $\pm$ 0.24                & 2.68 $\pm$ 0.20               & 6.70 $\pm$ 0.28                & 3.76 $\pm$ 0.23               & 15.85 $\pm$ 0.26               & 10.15 $\pm$ 0.12               \\
LightGCN       & 3.47 $\pm$ 0.41                & 3.13 $\pm$ 0.27               & \textbf{9.94 $\pm$ 0.23}                & \textbf{5.49 $\pm$ 0.14}               & 18.78 $\pm$ 0.10               & 12.25 $\pm$ 0.07              \\
HyperGNN       & \textbf{3.89 $\pm$ 0.34}       & \textbf{3.49 $\pm$ 0.30}      & 8.36 $\pm$ 0.41       & 4.34 $\pm$ 0.24      & \textbf{20.17 $\pm$ 0.28}      & \textbf{13.12 $\pm$ 0.21}     \\ \hline
FedRec         & 3.24 $\pm$ 0.46                & 2.62 $\pm$ 0.34               & 6.32 $\pm$ 0.39                & 3.50 $\pm$ 0.32               & 15.18 $\pm$ 0.35               & 9.27 $\pm$ 0.25               \\
FedGNN         & 3.38 $\pm$ 0.53                & 3.02 $\pm$ 0.36               & 7.40 $\pm$ 0.44                & 3.89 $\pm$ 0.33               & 17.29 $\pm$ 0.29               & 10.84 $\pm$ 0.13              \\ \hline
DMF            & 3.15 $\pm$ 0.20                & 2.52 $\pm$ 0.17               & 6.56 $\pm$ 0.36                & 3.62 $\pm$ 0.27               & 15.46 $\pm$ 0.22               & 9.77 $\pm$ 0.14               \\
DLGNN         & 3.30 $\pm$ 0.37                & 2.84 $\pm$ 0.28               & 7.73 $\pm$ 0.31                & 4.08 $\pm$ 0.24               & 18.50 $\pm$ 0.26               & 11.89 $\pm$ 0.19              \\ \hline
\textbf{\modelname} & \textbf{3.75 $\pm$ 0.41}       & \textbf{3.27 $\pm$ 0.25}      & \textbf{9.68 $\pm$ 0.35}       & \textbf{5.13 $\pm$ 0.21}      & \textbf{19.91 $\pm$ 0.31}      & \textbf{12.76 $\pm$ 0.18}     \\ \hline
\end{tabular}
\label{table_overall_comparison}
\end{table*}
This section conducts empirical studies on three public datasets to answer the following four research questions. 
\begin{itemize}[leftmargin=*]
\item \textbf{RQ1: }How does \modelname~perform, compared with the state-of-the-art methods (Section~\ref{subsec_performance_comparison})?
\item \textbf{RQ2: }How do existing methods perform under diverse demands for privacy protection (Section~\ref{subsec_generalization_research})? 
\item \textbf{RQ3:} What are the effects of different components in \modelname~(Section~\ref{subsec_ablation_study})?
\item \textbf{RQ4: }What are the impacts of different parameters on \modelname~(Section~\ref{subsec_parameter_analysis})?
\end{itemize}

\subsection{Experimental settings}
\label{subsec_experiment_settings}
\nosection{Datasets}
We conduct experiments on three public datasets, i.e., \textit{Flixster}, \textit{Book-crossing}, and \textit{Weeplaces}.
These datasets are collected with permissions and preprocessed through data anonymization.
\textit{Flixster} is a movie site where people meet others with similar movie tastes~\cite{GCMC}, and it provides user friendships and item relationships.
\textit{Book-crossing} contains a crawl from a book sharing community~\cite{exp_dataset_book_crossing}.
We connect users from the same area and build item relationships according to authors and publishers.
To ensure the dataset quality, we retain users and items with at least $18$ interactions~\cite{exp_dataset_10core_setting}.
\textit{Weeplaces} is a website that visualizes user check-in activities in location-based social networks~\cite{exp_dataset_weeplaces}.
We derive item relationships from their categories and remove users and items with less than $10$ interactions.
We show the statistics of these datasets after processing them in Table~\ref{exp_dataset_statics}.
For each dataset, we randomly select 80\% of the historical interactions of each user as the training dataset and treat the rest as the test set.
We randomly select $10\%$ of interactions in the training dataset as the validation set to tune hyperparameters.

\nosection{Comparison methods}
The comparison methods involve three types of recommendation methods, i.e., centralized methods, federated methods, and decentralized methods.

\nosection{(1) Centralized methods}
We implement three centralized recommendation methods, namely \textbf{NeuMF}~\cite{NeuMF}, \textbf{LightGCN}~\cite{LightGCN}, and \textbf{HyperGNN}~\cite{HyperGNN}.
NeuMF equips MLP to capture user preference on the item.
LightGCN and HyperGNN adopt graph neural networks.
The difference between LightGCN and HyperGNN is that LightGCN leverages pairwise relationships between users and items, whereas HyperGNN utilizes higher-order relationships. 

\nosection{(2) Federated methods}
We implement two federated recommendation methods, i.e., \textbf{FedRec}~\cite{FedREC} and \textbf{FedGNN}~\cite{FedGNN}.
FedRec is a federated matrix factorization (MF) method.
FedGNN models high-order user-item relationships and adds Gaussian noises to protect user privacy.

\nosection{(3) Decentralized methods}
We implement two decentralized recommendation methods named \textbf{DMF}~\cite{DMF} and \textbf{DLGNN}.
DMF is a decentralized MF method.
DLGNN is a decentralized implementation of LightGCN.
Each user trains his model collaboratively with his neighbors by sharing gradients in these two methods.

\textit{To give a fair comparison, we provide side information for all methods.}
(1) For MF methods, we leverage side information as regularization on embeddings, encouraging users and items to share similar embeddings with their neighbors~\cite{Soreg}.
(2) For GNN methods, we leverage side information to establish user-user and item-item relationships in the graph~\cite{KCGN}.

\nosection{Evaluation metrics}
For each user in the test set, we treat all the items that a user has not interacted with as negative.
Each method predicts each user's preferences on all items except the positive ones in the training set.
To evaluate the effectiveness of each model, we use two widely-adopted metrics, i.e., recall@K and ndcg@K~\cite{NeuMF, DGCF}.
In our experiments, we set $K=20$.
%

\nosection{Parameter settings}
\label{subsec_parameter_settings}
For all methods, we use SGD as the optimizer.
We search the learning rate in $[0.001,0.005,0.01]$, the L2 regularization coefficient in $[10^{-3}, 10^{-2}, 10^{-1}]$, and the embedding size in $\left[16,32,48,64\right]$.
For federated methods, we set the gradient clipping threshold $\delta$ to $0.1$ and privacy-preserving strength $\epsilon$ to $1$.
For our method, we set $\delta$ to $0.1$ and perturbation strength $\beta$ to $1$, which satisfies $(1.5, 0.54)$-RDP.
For GNN methods, we search the number of layers in $\left[1,2,3\right]$.
For NeuMF, we search the dropout in $[0.0,0.1,0.2,0.3,0.4,0.5]$.
For our method, we search the interest number $n_{i}$ in $[6, 12, 18]$ and interest dimension $d_{i}$ in $[5, 10, 15, 30]$, respectively.
Besides, we set sampling hops $H=4$ and sampling numbers $n_{u}=3$. 

We optimize all hyperparameters carefully through grid search for all methods to give a fair comparison.
We repeat the experiment with different random seeds five times and report their average results.
The consumed resources vary with the methods and hyperparameters.
We simulate a decentralized environment on a single device, establishing individual recommendation models for each user on that device.

\subsection{Overall performance comparisons (RQ1)}
\label{subsec_performance_comparison}
\begin{figure*}[t]
\subfigure{
\includegraphics[width=5.5cm]{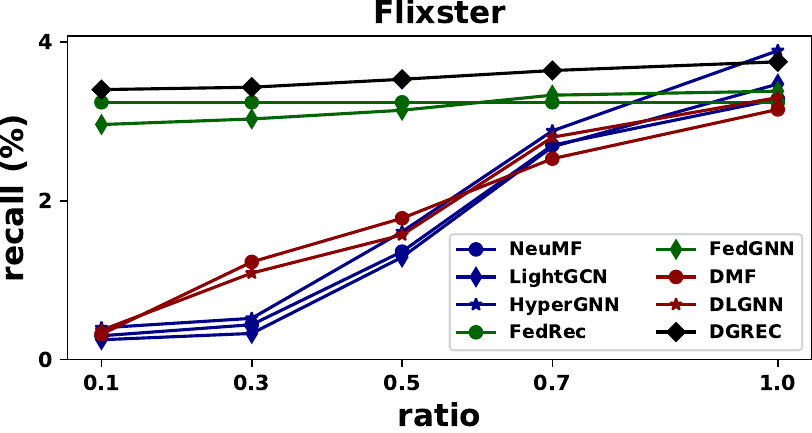}
\includegraphics[width=5.5cm]{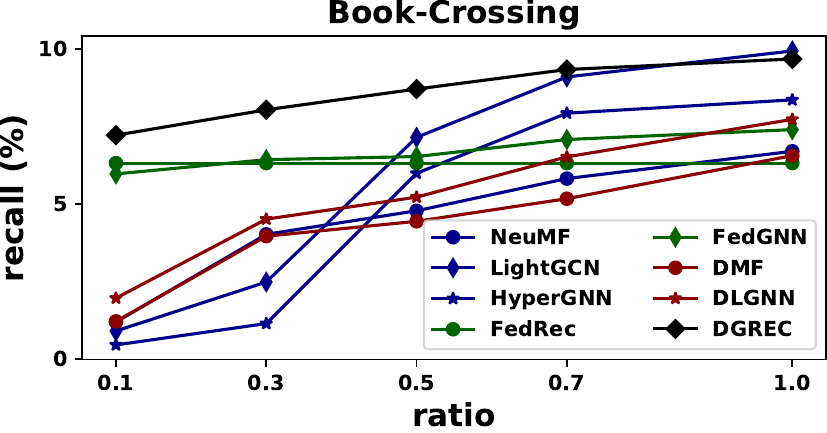}
\includegraphics[width=5.5cm]{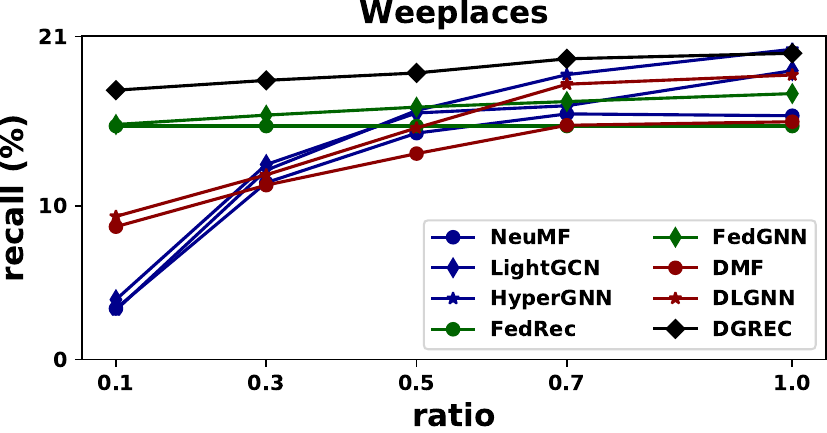}
}
\caption{Performance with different ratios of public interactions.}
\label{figure_public_ratio_comparison}
\end{figure*}
We compare each method in conventional recommendation scenarios where all user interactions are public.
We show the comparison results in Table~\ref{table_overall_comparison}.
From it, we observe that:
(1) Centralized methods achieve superior performance on all datasets as they can leverage all user interactions to make recommendations.
For protecting user privacy, federated and decentralized methods achieve competitive performance, demonstrating the feasibility of making accurate recommendations without violating user privacy.
(2) GNN methods generally outperform MF methods.
FedGNN, DLGNN, and \modelname~outperform the centralized MF method, i.e., NeuMF, on all datasets.
This result demonstrates the superiority of GNN in recommendations, as it can vastly enhance user and item representations.
(3) \modelname~achieves the best recommendation performance among all privacy-preserving methods.
It performs close to LightGCN and HyperGNN, which are two centralized GNNs.
These two results demonstrate that \modelname~can protect user privacy cost-effectively.

\begin{figure*}[t]
\subfigure[Performance on the \textit{Book-Crossing.}]{
\includegraphics[width=5.5cm]{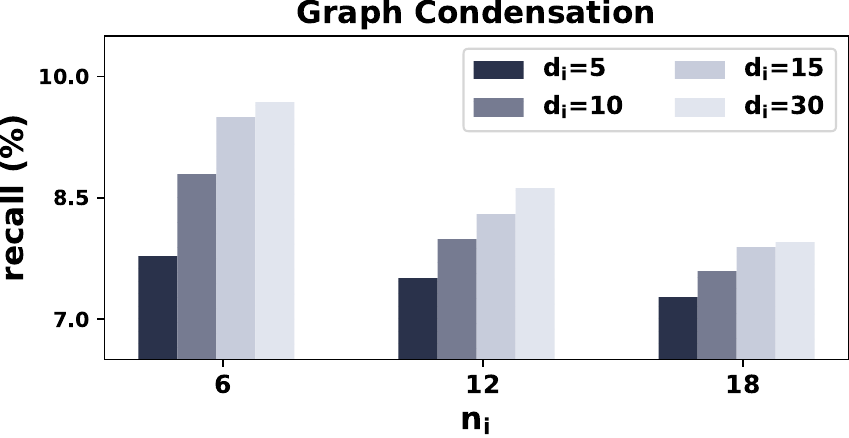}
\includegraphics[width=5.5cm]{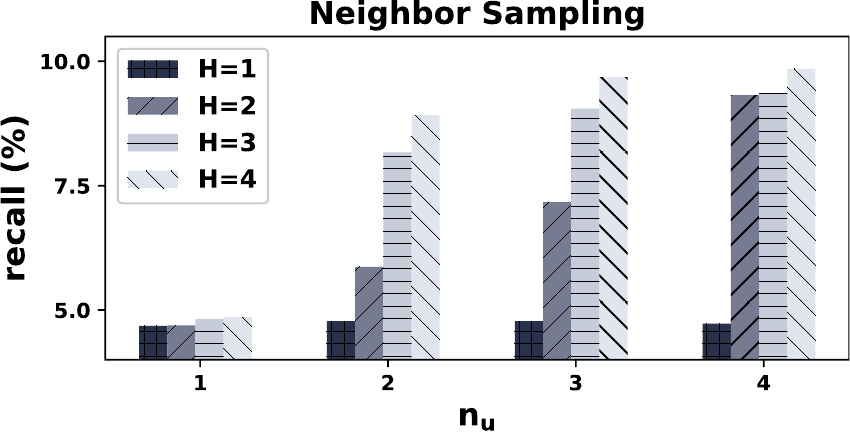}
\includegraphics[width=5.5cm]{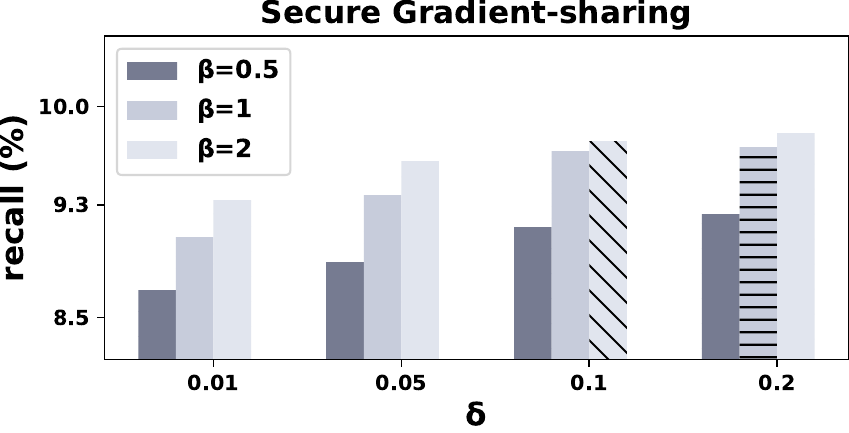}
}

\subfigure[Performance on the \textit{Weeplaces}.]{
\includegraphics[width=5.5cm]{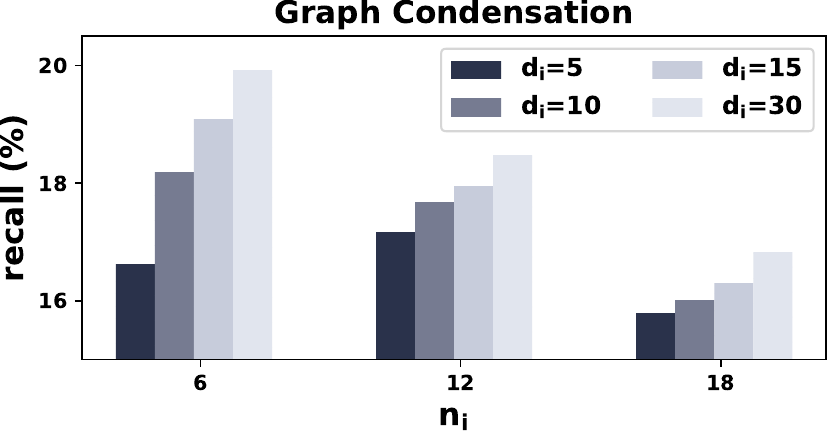}
\includegraphics[width=5.5cm]{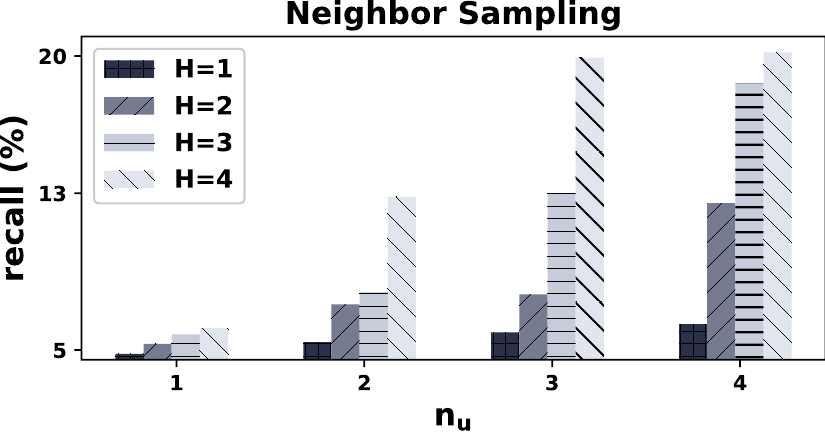}
\includegraphics[width=5.5cm]{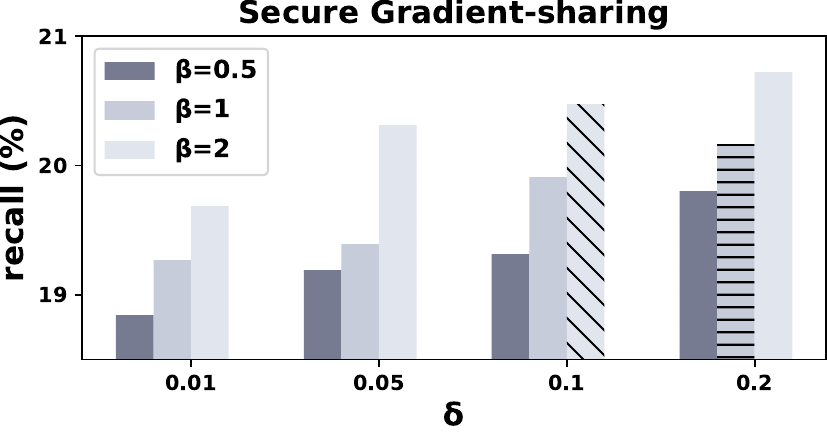}
}
\caption{Performance with different model parameters.}
\label{figure_parameter_analysis}
\end{figure*}

\subsection{Generalization research (RQ2)}
\label{subsec_generalization_research}
\begin{table}[t]
\centering
\caption{Performance on all variants of \modelname.}
\begin{tabular}{l|cc|cc}
\hline
                   & \multicolumn{2}{c|}{Book-Crossing} & \multicolumn{2}{c}{Weeplaces}   \\
                   & recall           & ndcg            & recall         & ndcg           \\ \hline
w/non-item graph     & 6.98             & 3.83            & 16.18          & 10.45          \\
w/non-neighbor   & 6.77             & 3.61            & 17.07          & 10.71          \\
w/non-item hypergraph      & 9.20             & 4.66            & 17.69          & 11.17          \\ \hline
w/non-attention      & 9.19             & 4.78            & 19.35          & 12.43          \\
w/non-pearson & 9.41             & 4.84            & 18.53          & 11.88          \\ \hline
w/non-sharing & 8.17             & 4.11            & 18.46          & 11.96          \\ \hline
\textbf{\modelname}     & \textbf{9.68}    & \textbf{5.13}   & \textbf{19.91} & \textbf{12.76} \\ \hline
\end{tabular}
\label{table_variants_performance}
\end{table}
We evaluate how different methods meet diverse demands for privacy protection.
Specifically, we set different publicized interaction ratios for each user, i.e., the proportion of publicized interactions to total interactions.
(1) Centralized and Decentralized methods can only utilize publicized interactions to build graphs and train models.
(2) Federated methods can only utilize publicized interactions to build the graph but all interactions to train models.
We show the comparison results in Figure~\ref{figure_public_ratio_comparison}.
From it, we observe that:
(1) All centralized and decentralized methods have performance degradation as publicized interaction ratio decreases.
In contrast, \modelname~and federated methods can keep certain recommendation performance.
(2) When we set publicized interaction ratio below $1$, \modelname~achieves consistent superiority on all datasets, demonstrating its ability to meet diverse demands for privacy protections.

\subsection{Ablation study (RQ3)}
\label{subsec_ablation_study}
We verify the effectiveness of different components in \modelname~on \textit{Book-Crossing} and \textit{Weeplaces}.
(1) \textbf{w/non-item graph}, \textbf{w/non-neighbor}, and \textbf{w/non-item hypergraph} are three variants of constructing item hypergraph. 
w/non-item graph averages embeddings of interacted items to model user preferences.
w/non-neighbor constructs the item hypergraph without neighbors' publicized interactions.
w/non-item hypergraph constructs an item-tag bipartite graph for each user.
(2) \textbf{w/non-attention} and \textbf{w/non-pearson} are two variants of modeling user preference.
w/non-attention replaces the interest attention with attention-pooling~\cite{attention}.
w/non-pearson removes the Pearson loss defined in Equation~\eqref{equation_pearson_loss}.
(3) \textbf{w/non-sharing} is the variant of secure gradient-sharing, which adds zero-mean Laplacian noise to protect user privacy in place of gradient encoding and decoding.

We show the comparison results of the ablation study in Table ~\ref{table_variants_performance}.
From it, we observe that:
(1) All proposed components are indispensable in \modelname.
Constructing an item graph and utilizing neighbors' publicized interactions are two dominant reasons for performance improvement.
The former contributes to distinguishing a user's core and temporary interests, whereas the latter leverages user behavior similarity.
(2) Constructing an item-tag bipartite graph cannot achieve competitive performance, especially for the datasets with the most user interactions, i.e., \textit{Weeplaces}.
Adding tag nodes hinders modeling user preference as tags are not related to the main task.  
(3) Utilizing interest attention and Pearson loss increases model performance, as these two methods can mitigate damage from noisy and redundant interests.
(4) Adding Laplacian noise to protect user privacy causes performance degradation, whereas secure gradient-sharing is noise-free and hence retains model performance.
%


\subsection{Parameter analyses (RQ4)}
\label{subsec_parameter_analysis}

\begin{figure}[t]
\centering
\includegraphics[width=7cm]{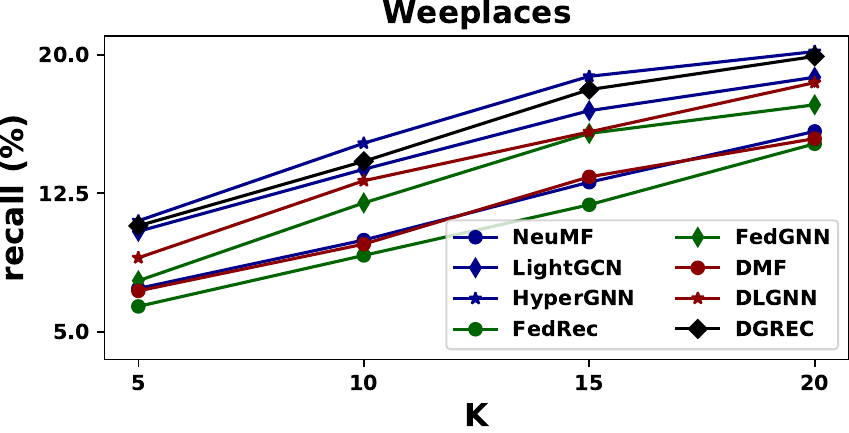}
\caption{Performance with different numbers of recommended items.}
\label{figure_pa_top_k}
\end{figure}

We first evaluate the number of recommended items $K$ for model performance on \textit{Weeplaces} and depict the results in Figure~\ref{figure_pa_top_k}.
We can observe that the performance of all models increases with $K$ and \modelname~achieves runner-up performances for all choices of $K$.

We then evaluate the impacts of different parameters on graph condensation, neighbor sampling, and secure gradient-sharing on \textit{Book-Crossing} and \textit{Weeplaces}.
We give the results of parameter analyses in Figure~\ref{figure_parameter_analysis}.
From it, we observe that:
(1) For graph condensation (number of interests $n_{i}$ and interest dimension $d_{i}$), model performance increases with $d_{i}$ but decreases with $n_{i}$.
A large interest dimension generally results in better interest representations, but redundant interests are bad for model convergence.
(2) For neighbor sampling (sampling hops $H$ and sampling numbers $n_{u}$), model performance increases with $H$ and $n_{u}$.
When we set $H$ to $1$, our gradient propagation is degraded to the gossip training mechanism.
Although users and their $1$-hop neighbors are more likely to share similar preferences, involving $1$-hop neighbors for model training cannot achieve competitive results.
The result motivates us to form a more extensive neighborhood to train recommendation models.
(3) For secure gradient-sharing mechanism (clip value $\delta$ and perturbation strength $\beta$), our model performance increases with $\delta$ and $\beta$.
However, setting larger $\delta$ and $\beta$ consumes a greater privacy budget, decreasing privacy-preserving strength.
In our scenario, setting $\delta=0.1$ and $\beta=1$ is sufficient to achieve good results.
%
%
\section{Conclusion}
\label{section_conclusion}
Under strict data protection rules, online platforms and large corporations are now precluded from amassing private user data to generate accurate recommendations. 
Despite this, the imperative need for recommender systems remains, particularly in mitigating the issue of information overload.
The development and refinement of privacy-preserving recommender systems receive increasing attention.

In this paper, we propose \modelname~for privacy-preserving recommendation, which achieves superior recommendation performance and provides strong privacy protection.
\modelname~allows users to publicize their interactions freely, which includes three stages, i.e., graph construction, local gradient calculation, and global gradient passing.
In future work, we will extend our research to untrusted environments, where other clients may be semi-honest and malicious, and study cryptography techniques, e.g., zero-knowledge proof, to solve the issue.
We will also investigate defence methods against poisoning attacks.
%

\section*{Acknowledgment}
This work is supported by National Key R\&D Program of China (2022YFB4501500, 2022YFB4501504).

\bibliographystyle{ACM-Reference-Format}
\balance
\bibliography{reference}


\begin{thebibliography}{60}


\ifx \showCODEN    \undefined \def \showCODEN     #1{\unskip}     \fi
\ifx \showDOI      \undefined \def \showDOI       #1{#1}\fi
\ifx \showISBNx    \undefined \def \showISBNx     #1{\unskip}     \fi
\ifx \showISBNxiii \undefined \def \showISBNxiii  #1{\unskip}     \fi
\ifx \showISSN     \undefined \def \showISSN      #1{\unskip}     \fi
\ifx \showLCCN     \undefined \def \showLCCN      #1{\unskip}     \fi
\ifx \shownote     \undefined \def \shownote      #1{#1}          \fi
\ifx \showarticletitle \undefined \def \showarticletitle #1{#1}   \fi
\ifx \showURL      \undefined \def \showURL       {\relax}        \fi
\providecommand\bibfield[2]{#2}
\providecommand\bibinfo[2]{#2}
\providecommand\natexlab[1]{#1}
\providecommand\showeprint[2][]{arXiv:#2}

\bibitem[Archer et~al\mbox{.}(2017)]%
        {HE3}
\bibfield{author}{\bibinfo{person}{David Archer}, \bibinfo{person}{Lily Chen},
  \bibinfo{person}{Jung~Hee Cheon}, \bibinfo{person}{Ran Gilad-Bachrach},
  \bibinfo{person}{Roger~A Hallman}, \bibinfo{person}{Zhicong Huang},
  \bibinfo{person}{Xiaoqian Jiang}, \bibinfo{person}{Ranjit Kumaresan},
  \bibinfo{person}{Bradley~A Malin}, \bibinfo{person}{Heidi Sofia},
  {et~al\mbox{.}}} \bibinfo{year}{2017}\natexlab{}.
\newblock \showarticletitle{Applications of homomorphic encryption}. In
  \bibinfo{booktitle}{\emph{Crypto Standardization Workshop, Microsoft
  Research}}, Vol.~\bibinfo{volume}{14}.
\newblock


\bibitem[Ayache and Rouayheb(2021)]%
        {convergence_proof}
\bibfield{author}{\bibinfo{person}{Ghadir Ayache} {and}
  \bibinfo{person}{Salim~El Rouayheb}.} \bibinfo{year}{2021}\natexlab{}.
\newblock \showarticletitle{Private Weighted Random Walk Stochastic Gradient
  Descent}.
\newblock \bibinfo{journal}{\emph{{IEEE} J. Sel. Areas Inf. Theory}}
  \bibinfo{volume}{2}, \bibinfo{number}{1} (\bibinfo{year}{2021}),
  \bibinfo{pages}{452--463}.
\newblock
\urldef\tempurl%
\url{https://doi.org/10.1109/JSAIT.2021.3052975}
\showDOI{\tempurl}


\bibitem[Bai et~al\mbox{.}(2021)]%
        {HyperGNN}
\bibfield{author}{\bibinfo{person}{Song Bai}, \bibinfo{person}{Feihu Zhang},
  {and} \bibinfo{person}{Philip H.~S. Torr}.} \bibinfo{year}{2021}\natexlab{}.
\newblock \showarticletitle{Hypergraph convolution and hypergraph attention}.
\newblock \bibinfo{journal}{\emph{Pattern Recognit.}}  \bibinfo{volume}{110}
  (\bibinfo{year}{2021}), \bibinfo{pages}{107637}.
\newblock
\urldef\tempurl%
\url{https://doi.org/10.1016/j.patcog.2020.107637}
\showDOI{\tempurl}


\bibitem[Benson et~al\mbox{.}(2016)]%
        {hypergraph}
\bibfield{author}{\bibinfo{person}{Austin~R Benson}, \bibinfo{person}{David~F
  Gleich}, {and} \bibinfo{person}{Jure Leskovec}.}
  \bibinfo{year}{2016}\natexlab{}.
\newblock \showarticletitle{Higher-order organization of complex networks}.
\newblock \bibinfo{journal}{\emph{Science}} \bibinfo{volume}{353},
  \bibinfo{number}{6295} (\bibinfo{year}{2016}), \bibinfo{pages}{163--166}.
\newblock


\bibitem[Bianchi et~al\mbox{.}(2020)]%
        {min_cut_pool}
\bibfield{author}{\bibinfo{person}{Filippo~Maria Bianchi},
  \bibinfo{person}{Daniele Grattarola}, {and} \bibinfo{person}{Cesare Alippi}.}
  \bibinfo{year}{2020}\natexlab{}.
\newblock \showarticletitle{Spectral clustering with graph neural networks for
  graph pooling}. In \bibinfo{booktitle}{\emph{ICML}}.
  \bibinfo{pages}{874--883}.
\newblock


\bibitem[Cangea et~al\mbox{.}(2018)]%
        {gnn_topk_pooling2}
\bibfield{author}{\bibinfo{person}{C{\u{a}}t{\u{a}}lina Cangea},
  \bibinfo{person}{Petar Veli{\v{c}}kovi{\'c}}, \bibinfo{person}{Nikola
  Jovanovi{\'c}}, \bibinfo{person}{Thomas Kipf}, {and} \bibinfo{person}{Pietro
  Li{\`o}}.} \bibinfo{year}{2018}\natexlab{}.
\newblock \showarticletitle{Towards sparse hierarchical graph classifiers}.
\newblock \bibinfo{journal}{\emph{arXiv preprint arXiv:1811.01287}}
  (\bibinfo{year}{2018}).
\newblock


\bibitem[Chen et~al\mbox{.}(2018)]%
        {DMF}
\bibfield{author}{\bibinfo{person}{Chaochao Chen}, \bibinfo{person}{Ziqi Liu},
  \bibinfo{person}{Peilin Zhao}, \bibinfo{person}{Jun Zhou}, {and}
  \bibinfo{person}{Xiaolong Li}.} \bibinfo{year}{2018}\natexlab{}.
\newblock \showarticletitle{Privacy Preserving Point-of-Interest Recommendation
  Using Decentralized Matrix Factorization}. In
  \bibinfo{booktitle}{\emph{AAAI}},
  \bibfield{editor}{\bibinfo{person}{Sheila~A. McIlraith} {and}
  \bibinfo{person}{Kilian~Q. Weinberger}} (Eds.). \bibinfo{pages}{257--264}.
\newblock


\bibitem[Chen et~al\mbox{.}(2022)]%
        {chen2022differential}
\bibfield{author}{\bibinfo{person}{Chaochao Chen}, \bibinfo{person}{Huiwen Wu},
  \bibinfo{person}{Jiajie Su}, \bibinfo{person}{Lingjuan Lyu},
  \bibinfo{person}{Xiaolin Zheng}, {and} \bibinfo{person}{Li Wang}.}
  \bibinfo{year}{2022}\natexlab{}.
\newblock \showarticletitle{Differential private knowledge transfer for
  privacy-preserving cross-domain recommendation}. In
  \bibinfo{booktitle}{\emph{Proceedings of the ACM Web Conference 2022}}.
  \bibinfo{pages}{1455--1465}.
\newblock


\bibitem[Cheon et~al\mbox{.}(2017)]%
        {HE2}
\bibfield{author}{\bibinfo{person}{Jung~Hee Cheon}, \bibinfo{person}{Andrey
  Kim}, \bibinfo{person}{Miran Kim}, {and} \bibinfo{person}{Yongsoo Song}.}
  \bibinfo{year}{2017}\natexlab{}.
\newblock \showarticletitle{Homomorphic encryption for arithmetic of
  approximate numbers}. In \bibinfo{booktitle}{\emph{ASIACRYPT}}.
  \bibinfo{pages}{409--437}.
\newblock


\bibitem[Daily et~al\mbox{.}(2018)]%
        {gossip}
\bibfield{author}{\bibinfo{person}{Jeff Daily}, \bibinfo{person}{Abhinav
  Vishnu}, \bibinfo{person}{Charles Siegel}, \bibinfo{person}{Thomas Warfel},
  {and} \bibinfo{person}{Vinay Amatya}.} \bibinfo{year}{2018}\natexlab{}.
\newblock \showarticletitle{GossipGraD: Scalable Deep Learning using Gossip
  Communication based Asynchronous Gradient Descent}.
\newblock \bibinfo{journal}{\emph{CoRR}}  \bibinfo{volume}{abs/1803.05880}
  (\bibinfo{year}{2018}).
\newblock
\showeprint[arXiv]{1803.05880}
\urldef\tempurl%
\url{http://arxiv.org/abs/1803.05880}
\showURL{%
\tempurl}


\bibitem[Dijk et~al\mbox{.}(2010)]%
        {HE1}
\bibfield{author}{\bibinfo{person}{Marten~van Dijk}, \bibinfo{person}{Craig
  Gentry}, \bibinfo{person}{Shai Halevi}, {and} \bibinfo{person}{Vinod
  Vaikuntanathan}.} \bibinfo{year}{2010}\natexlab{}.
\newblock \showarticletitle{Fully homomorphic encryption over the integers}. In
  \bibinfo{booktitle}{\emph{EUROCRYPT}}. \bibinfo{pages}{24--43}.
\newblock


\bibitem[Ding et~al\mbox{.}(2017)]%
        {motivation_one_bit_encoder}
\bibfield{author}{\bibinfo{person}{Bolin Ding}, \bibinfo{person}{Janardhan
  Kulkarni}, {and} \bibinfo{person}{Sergey Yekhanin}.}
  \bibinfo{year}{2017}\natexlab{}.
\newblock \showarticletitle{Collecting Telemetry Data Privately}. In
  \bibinfo{booktitle}{\emph{NIPS}}, \bibfield{editor}{\bibinfo{person}{Isabelle
  Guyon}, \bibinfo{person}{Ulrike von Luxburg}, \bibinfo{person}{Samy Bengio},
  \bibinfo{person}{Hanna~M. Wallach}, \bibinfo{person}{Rob Fergus},
  \bibinfo{person}{S.~V.~N. Vishwanathan}, {and} \bibinfo{person}{Roman
  Garnett}} (Eds.). \bibinfo{pages}{3571--3580}.
\newblock


\bibitem[Duvenaud et~al\mbox{.}(2015)]%
        {gnn_summing_pooling}
\bibfield{author}{\bibinfo{person}{David~K Duvenaud}, \bibinfo{person}{Dougal
  Maclaurin}, \bibinfo{person}{Jorge Iparraguirre}, \bibinfo{person}{Rafael
  Bombarell}, \bibinfo{person}{Timothy Hirzel}, \bibinfo{person}{Al{\'a}n
  Aspuru-Guzik}, {and} \bibinfo{person}{Ryan~P Adams}.}
  \bibinfo{year}{2015}\natexlab{}.
\newblock \showarticletitle{Convolutional networks on graphs for learning
  molecular fingerprints}.
\newblock \bibinfo{journal}{\emph{NIPS}}  \bibinfo{volume}{28}
  (\bibinfo{year}{2015}).
\newblock


\bibitem[Dwork et~al\mbox{.}(2006a)]%
        {dp1}
\bibfield{author}{\bibinfo{person}{Cynthia Dwork}, \bibinfo{person}{Krishnaram
  Kenthapadi}, \bibinfo{person}{Frank McSherry}, \bibinfo{person}{Ilya
  Mironov}, {and} \bibinfo{person}{Moni Naor}.}
  \bibinfo{year}{2006}\natexlab{a}.
\newblock \showarticletitle{Our data, ourselves: Privacy via distributed noise
  generation}. In \bibinfo{booktitle}{\emph{EUROCRYPT}}.
  \bibinfo{pages}{486--503}.
\newblock


\bibitem[Dwork et~al\mbox{.}(2006b)]%
        {dp0}
\bibfield{author}{\bibinfo{person}{Cynthia Dwork}, \bibinfo{person}{Frank
  McSherry}, \bibinfo{person}{Kobbi Nissim}, {and} \bibinfo{person}{Adam
  Smith}.} \bibinfo{year}{2006}\natexlab{b}.
\newblock \showarticletitle{Calibrating noise to sensitivity in private data
  analysis}. In \bibinfo{booktitle}{\emph{TCC}}. \bibinfo{pages}{265--284}.
\newblock


\bibitem[Dwork et~al\mbox{.}(2006c)]%
        {sequential_composition}
\bibfield{author}{\bibinfo{person}{Cynthia Dwork}, \bibinfo{person}{Frank
  McSherry}, \bibinfo{person}{Kobbi Nissim}, {and} \bibinfo{person}{Adam
  Smith}.} \bibinfo{year}{2006}\natexlab{c}.
\newblock \showarticletitle{Calibrating noise to sensitivity in private data
  analysis}. In \bibinfo{booktitle}{\emph{TCC}}. Springer,
  \bibinfo{pages}{265--284}.
\newblock


\bibitem[Feng et~al\mbox{.}(2019)]%
        {first_hypergcn}
\bibfield{author}{\bibinfo{person}{Yifan Feng}, \bibinfo{person}{Haoxuan You},
  \bibinfo{person}{Zizhao Zhang}, \bibinfo{person}{Rongrong Ji}, {and}
  \bibinfo{person}{Yue Gao}.} \bibinfo{year}{2019}\natexlab{}.
\newblock \showarticletitle{Hypergraph neural networks}. In
  \bibinfo{booktitle}{\emph{AAAI}}, Vol.~\bibinfo{volume}{33}.
  \bibinfo{pages}{3558--3565}.
\newblock


\bibitem[Gao and Ji(2019)]%
        {gnn_topk_pooling}
\bibfield{author}{\bibinfo{person}{Hongyang Gao} {and}
  \bibinfo{person}{Shuiwang Ji}.} \bibinfo{year}{2019}\natexlab{}.
\newblock \showarticletitle{Graph u-nets}. In \bibinfo{booktitle}{\emph{ICML}}.
  \bibinfo{pages}{2083--2092}.
\newblock


\bibitem[Gao et~al\mbox{.}(2022)]%
        {DLGNN}
\bibfield{author}{\bibinfo{person}{Huiguo Gao}, \bibinfo{person}{Mengyuan Lee},
  \bibinfo{person}{Guanding Yu}, {and} \bibinfo{person}{Zhaolin Zhou}.}
  \bibinfo{year}{2022}\natexlab{}.
\newblock \showarticletitle{A Graph Neural Network Based Decentralized Learning
  Scheme}.
\newblock \bibinfo{journal}{\emph{Sensors}} \bibinfo{volume}{22},
  \bibinfo{number}{3} (\bibinfo{year}{2022}), \bibinfo{pages}{1030}.
\newblock
\urldef\tempurl%
\url{https://doi.org/10.3390/s22031030}
\showDOI{\tempurl}


\bibitem[Gentry(2009)]%
        {homomorphic_encryption}
\bibfield{author}{\bibinfo{person}{Craig Gentry}.}
  \bibinfo{year}{2009}\natexlab{}.
\newblock \showarticletitle{Fully homomorphic encryption using ideal lattices}.
  In \bibinfo{booktitle}{\emph{STOC}},
  \bibfield{editor}{\bibinfo{person}{Michael Mitzenmacher}} (Ed.).
  \bibinfo{publisher}{{ACM}}, \bibinfo{pages}{169--178}.
\newblock
\urldef\tempurl%
\url{https://doi.org/10.1145/1536414.1536440}
\showDOI{\tempurl}


\bibitem[Goldreich(1998)]%
        {MPC1}
\bibfield{author}{\bibinfo{person}{Oded Goldreich}.}
  \bibinfo{year}{1998}\natexlab{}.
\newblock \showarticletitle{Secure multi-party computation}.
\newblock \bibinfo{journal}{\emph{Manuscript. Preliminary version}}
  \bibinfo{volume}{78} (\bibinfo{year}{1998}), \bibinfo{pages}{110}.
\newblock


\bibitem[Hamilton et~al\mbox{.}(2017)]%
        {graphsage}
\bibfield{author}{\bibinfo{person}{William~L. Hamilton},
  \bibinfo{person}{Zhitao Ying}, {and} \bibinfo{person}{Jure Leskovec}.}
  \bibinfo{year}{2017}\natexlab{}.
\newblock \showarticletitle{Inductive Representation Learning on Large Graphs}.
  In \bibinfo{booktitle}{\emph{NIPS}}. \bibinfo{pages}{1024--1034}.
\newblock


\bibitem[He et~al\mbox{.}(2021)]%
        {spread_gnn}
\bibfield{author}{\bibinfo{person}{Chaoyang He}, \bibinfo{person}{Emir Ceyani},
  \bibinfo{person}{Keshav Balasubramanian}, \bibinfo{person}{Murali Annavaram},
  {and} \bibinfo{person}{Salman Avestimehr}.} \bibinfo{year}{2021}\natexlab{}.
\newblock \showarticletitle{SpreadGNN: Serverless Multi-task Federated Learning
  for Graph Neural Networks}.
\newblock \bibinfo{journal}{\emph{CoRR}}  \bibinfo{volume}{abs/2106.02743}
  (\bibinfo{year}{2021}).
\newblock
\showeprint[arXiv]{2106.02743}
\urldef\tempurl%
\url{https://arxiv.org/abs/2106.02743}
\showURL{%
\tempurl}


\bibitem[He and McAuley(2016)]%
        {exp_dataset_10core_setting}
\bibfield{author}{\bibinfo{person}{Ruining He} {and} \bibinfo{person}{Julian~J.
  McAuley}.} \bibinfo{year}{2016}\natexlab{}.
\newblock \showarticletitle{{VBPR:} Visual Bayesian Personalized Ranking from
  Implicit Feedback}. In \bibinfo{booktitle}{\emph{AAAI}}.
  \bibinfo{pages}{144--150}.
\newblock


\bibitem[He et~al\mbox{.}(2020)]%
        {LightGCN}
\bibfield{author}{\bibinfo{person}{Xiangnan He}, \bibinfo{person}{Kuan Deng},
  \bibinfo{person}{Xiang Wang}, \bibinfo{person}{Yan Li},
  \bibinfo{person}{Yong{-}Dong Zhang}, {and} \bibinfo{person}{Meng Wang}.}
  \bibinfo{year}{2020}\natexlab{}.
\newblock \showarticletitle{LightGCN: Simplifying and Powering Graph
  Convolution Network for Recommendation}. In
  \bibinfo{booktitle}{\emph{SIGIR}}. \bibinfo{publisher}{{ACM}},
  \bibinfo{pages}{639--648}.
\newblock
\urldef\tempurl%
\url{https://doi.org/10.1145/3397271.3401063}
\showDOI{\tempurl}


\bibitem[He et~al\mbox{.}(2017)]%
        {NeuMF}
\bibfield{author}{\bibinfo{person}{Xiangnan He}, \bibinfo{person}{Lizi Liao},
  \bibinfo{person}{Hanwang Zhang}, \bibinfo{person}{Liqiang Nie},
  \bibinfo{person}{Xia Hu}, {and} \bibinfo{person}{Tat{-}Seng Chua}.}
  \bibinfo{year}{2017}\natexlab{}.
\newblock \showarticletitle{Neural Collaborative Filtering}. In
  \bibinfo{booktitle}{\emph{WWW}}. \bibinfo{publisher}{{ACM}},
  \bibinfo{pages}{173--182}.
\newblock
\urldef\tempurl%
\url{https://doi.org/10.1145/3038912.3052569}
\showDOI{\tempurl}


\bibitem[Huang et~al\mbox{.}(2021)]%
        {KCGN}
\bibfield{author}{\bibinfo{person}{Chao Huang}, \bibinfo{person}{Huance Xu},
  \bibinfo{person}{Yong Xu}, \bibinfo{person}{Peng Dai},
  \bibinfo{person}{Lianghao Xia}, \bibinfo{person}{Mengyin Lu},
  \bibinfo{person}{Liefeng Bo}, \bibinfo{person}{Hao Xing},
  \bibinfo{person}{Xiaoping Lai}, {and} \bibinfo{person}{Yanfang Ye}.}
  \bibinfo{year}{2021}\natexlab{}.
\newblock \showarticletitle{Knowledge-aware Coupled Graph Neural Network for
  Social Recommendation}. In \bibinfo{booktitle}{\emph{AAAI 2021}}.
  \bibinfo{publisher}{AAAI Press}, \bibinfo{pages}{4115--4122}.
\newblock
\urldef\tempurl%
\url{https://ojs.aaai.org/index.php/AAAI/article/view/16533}
\showURL{%
\tempurl}


\bibitem[Jin and Chen(2022)]%
        {GWGNN}
\bibfield{author}{\bibinfo{person}{Hongwei Jin} {and} \bibinfo{person}{Xun
  Chen}.} \bibinfo{year}{2022}\natexlab{}.
\newblock \showarticletitle{Gromov-Wasserstein Discrepancy with Local
  Differential Privacy for Distributed Structural Graphs}.
\newblock \bibinfo{journal}{\emph{CoRR}}  \bibinfo{volume}{abs/2202.00808}
  (\bibinfo{year}{2022}).
\newblock
\showeprint[arXiv]{2202.00808}
\urldef\tempurl%
\url{https://arxiv.org/abs/2202.00808}
\showURL{%
\tempurl}


\bibitem[Kipf and Welling(2017)]%
        {GCN}
\bibfield{author}{\bibinfo{person}{Thomas~N. Kipf} {and} \bibinfo{person}{Max
  Welling}.} \bibinfo{year}{2017}\natexlab{}.
\newblock \showarticletitle{Semi-Supervised Classification with Graph
  Convolutional Networks}. In \bibinfo{booktitle}{\emph{ICLR}}.
\newblock
\urldef\tempurl%
\url{https://openreview.net/forum?id=SJU4ayYgl}
\showURL{%
\tempurl}


\bibitem[Krasanakis et~al\mbox{.}(2022)]%
        {p2pGNN}
\bibfield{author}{\bibinfo{person}{Emmanouil Krasanakis},
  \bibinfo{person}{Symeon Papadopoulos}, {and} \bibinfo{person}{Ioannis
  Kompatsiaris}.} \bibinfo{year}{2022}\natexlab{}.
\newblock \showarticletitle{p2pGNN: {A} Decentralized Graph Neural Network for
  Node Classification in Peer-to-Peer Networks}.
\newblock \bibinfo{journal}{\emph{IEEE Access}}  \bibinfo{volume}{10}
  (\bibinfo{year}{2022}), \bibinfo{pages}{34755--34765}.
\newblock
\urldef\tempurl%
\url{https://doi.org/10.1109/ACCESS.2022.3159688}
\showDOI{\tempurl}


\bibitem[Li and Milenkovic(2017)]%
        {hypergraph_define2}
\bibfield{author}{\bibinfo{person}{Pan Li} {and} \bibinfo{person}{Olgica
  Milenkovic}.} \bibinfo{year}{2017}\natexlab{}.
\newblock \showarticletitle{Inhomogeneous hypergraph clustering with
  applications}.
\newblock \bibinfo{journal}{\emph{NIPS}}  \bibinfo{volume}{30}
  (\bibinfo{year}{2017}).
\newblock


\bibitem[Lian et~al\mbox{.}(2017)]%
        {central_decentral_system}
\bibfield{author}{\bibinfo{person}{Xiangru Lian}, \bibinfo{person}{Ce Zhang},
  \bibinfo{person}{Huan Zhang}, \bibinfo{person}{Cho{-}Jui Hsieh},
  \bibinfo{person}{Wei Zhang}, {and} \bibinfo{person}{Ji Liu}.}
  \bibinfo{year}{2017}\natexlab{}.
\newblock \showarticletitle{Can Decentralized Algorithms Outperform Centralized
  Algorithms? {A} Case Study for Decentralized Parallel Stochastic Gradient
  Descent}. In \bibinfo{booktitle}{\emph{NIPS}}. \bibinfo{pages}{5330--5340}.
\newblock


\bibitem[Lian et~al\mbox{.}(2018)]%
        {DPSGD}
\bibfield{author}{\bibinfo{person}{Xiangru Lian}, \bibinfo{person}{Wei Zhang},
  \bibinfo{person}{Ce Zhang}, {and} \bibinfo{person}{Ji Liu}.}
  \bibinfo{year}{2018}\natexlab{}.
\newblock \showarticletitle{Asynchronous Decentralized Parallel Stochastic
  Gradient Descent}. In \bibinfo{booktitle}{\emph{ICML}},
  \bibfield{editor}{\bibinfo{person}{Jennifer~G. Dy} {and}
  \bibinfo{person}{Andreas Krause}} (Eds.), Vol.~\bibinfo{volume}{80}.
  \bibinfo{publisher}{{PMLR}}, \bibinfo{pages}{3049--3058}.
\newblock


\bibitem[Liao et~al\mbox{.}(2023)]%
        {DBLP:conf/aaai/LiaoLZY023}
\bibfield{author}{\bibinfo{person}{Xinting Liao}, \bibinfo{person}{Weiming
  Liu}, \bibinfo{person}{Xiaolin Zheng}, \bibinfo{person}{Binhui Yao}, {and}
  \bibinfo{person}{Chaochao Chen}.} \bibinfo{year}{2023}\natexlab{}.
\newblock \showarticletitle{PPGenCDR: {A} Stable and Robust Framework for
  Privacy-Preserving Cross-Domain Recommendation}. In
  \bibinfo{booktitle}{\emph{AAAI}}, \bibfield{editor}{\bibinfo{person}{Brian
  Williams}, \bibinfo{person}{Yiling Chen}, {and} \bibinfo{person}{Jennifer
  Neville}} (Eds.). \bibinfo{publisher}{{AAAI} Press},
  \bibinfo{pages}{4453--4461}.
\newblock


\bibitem[Lin et~al\mbox{.}(2021)]%
        {FedREC}
\bibfield{author}{\bibinfo{person}{Guanyu Lin}, \bibinfo{person}{Feng Liang},
  \bibinfo{person}{Weike Pan}, {and} \bibinfo{person}{Zhong Ming}.}
  \bibinfo{year}{2021}\natexlab{}.
\newblock \showarticletitle{FedRec: Federated Recommendation With Explicit
  Feedback}.
\newblock \bibinfo{journal}{\emph{{IEEE} Intell. Syst.}} \bibinfo{volume}{36},
  \bibinfo{number}{5} (\bibinfo{year}{2021}), \bibinfo{pages}{21--30}.
\newblock
\urldef\tempurl%
\url{https://doi.org/10.1109/MIS.2020.3017205}
\showDOI{\tempurl}


\bibitem[Liu et~al\mbox{.}(2018)]%
        {gossip_2}
\bibfield{author}{\bibinfo{person}{Yang Liu}, \bibinfo{person}{Ji Liu}, {and}
  \bibinfo{person}{Tamer Basar}.} \bibinfo{year}{2018}\natexlab{}.
\newblock \showarticletitle{Gossip Gradient Descent}. In
  \bibinfo{booktitle}{\emph{AAMAS}}. \bibinfo{pages}{1995--1997}.
\newblock
\urldef\tempurl%
\url{http://dl.acm.org/citation.cfm?id=3238049}
\showURL{%
\tempurl}


\bibitem[Liu et~al\mbox{.}(2021)]%
        {FeSoG}
\bibfield{author}{\bibinfo{person}{Zhiwei Liu}, \bibinfo{person}{Liangwei
  Yang}, \bibinfo{person}{Ziwei Fan}, \bibinfo{person}{Hao Peng}, {and}
  \bibinfo{person}{Philip~S. Yu}.} \bibinfo{year}{2021}\natexlab{}.
\newblock \showarticletitle{Federated Social Recommendation with Graph Neural
  Network}.
\newblock \bibinfo{journal}{\emph{CoRR}}  \bibinfo{volume}{abs/2111.10778}
  (\bibinfo{year}{2021}).
\newblock
\showeprint[arXiv]{2111.10778}
\urldef\tempurl%
\url{https://arxiv.org/abs/2111.10778}
\showURL{%
\tempurl}


\bibitem[Ma et~al\mbox{.}(2011)]%
        {Soreg}
\bibfield{author}{\bibinfo{person}{Hao Ma}, \bibinfo{person}{Dengyong Zhou},
  \bibinfo{person}{Chao Liu}, \bibinfo{person}{Michael~R. Lyu}, {and}
  \bibinfo{person}{Irwin King}.} \bibinfo{year}{2011}\natexlab{}.
\newblock \showarticletitle{Recommender systems with social regularization}. In
  \bibinfo{booktitle}{\emph{WSDM 2011}},
  \bibfield{editor}{\bibinfo{person}{Irwin King}, \bibinfo{person}{Wolfgang
  Nejdl}, {and} \bibinfo{person}{Hang Li}} (Eds.). \bibinfo{publisher}{{ACM}},
  \bibinfo{pages}{287--296}.
\newblock
\urldef\tempurl%
\url{https://doi.org/10.1145/1935826.1935877}
\showDOI{\tempurl}


\bibitem[Mironov(2017)]%
        {renyi_differential_privacy}
\bibfield{author}{\bibinfo{person}{Ilya Mironov}.}
  \bibinfo{year}{2017}\natexlab{}.
\newblock \showarticletitle{R{\'{e}}nyi Differential Privacy}. In
  \bibinfo{booktitle}{\emph{IEEE CSF}}. \bibinfo{publisher}{{IEEE} Computer
  Society}, \bibinfo{pages}{263--275}.
\newblock
\urldef\tempurl%
\url{https://doi.org/10.1109/CSF.2017.11}
\showDOI{\tempurl}


\bibitem[Pei et~al\mbox{.}(2021)]%
        {decentralized_federated_GNN}
\bibfield{author}{\bibinfo{person}{Yang Pei}, \bibinfo{person}{Renxin Mao},
  \bibinfo{person}{Yang Liu}, \bibinfo{person}{Chaoran Chen},
  \bibinfo{person}{Shifeng Xu}, \bibinfo{person}{Feng Qiang}, {and}
  \bibinfo{person}{Blue~Elephant Tech}.} \bibinfo{year}{2021}\natexlab{}.
\newblock \showarticletitle{Decentralized Federated Graph Neural Networks}. In
  \bibinfo{booktitle}{\emph{IJCAI}}.
\newblock


\bibitem[Ren et~al\mbox{.}(2018)]%
        {DP_current_Fed}
\bibfield{author}{\bibinfo{person}{Xuebin Ren}, \bibinfo{person}{Chia{-}Mu Yu},
  \bibinfo{person}{Weiren Yu}, \bibinfo{person}{Shusen Yang},
  \bibinfo{person}{Xinyu Yang}, \bibinfo{person}{Julie~A. McCann}, {and}
  \bibinfo{person}{Philip~S. Yu}.} \bibinfo{year}{2018}\natexlab{}.
\newblock \showarticletitle{LoPub: High-Dimensional Crowdsourced Data
  Publication With Local Differential Privacy}.
\newblock \bibinfo{journal}{\emph{{IEEE} Trans}} \bibinfo{volume}{13},
  \bibinfo{number}{9} (\bibinfo{year}{2018}), \bibinfo{pages}{2151--2166}.
\newblock
\urldef\tempurl%
\url{https://doi.org/10.1109/TIFS.2018.2812146}
\showDOI{\tempurl}


\bibitem[Rendle et~al\mbox{.}(2009)]%
        {BPRLoss}
\bibfield{author}{\bibinfo{person}{Steffen Rendle}, \bibinfo{person}{Christoph
  Freudenthaler}, \bibinfo{person}{Zeno Gantner}, {and} \bibinfo{person}{Lars
  Schmidt{-}Thieme}.} \bibinfo{year}{2009}\natexlab{}.
\newblock \showarticletitle{{BPR:} Bayesian Personalized Ranking from Implicit
  Feedback}. In \bibinfo{booktitle}{\emph{UAI}}. \bibinfo{pages}{452--461}.
\newblock


\bibitem[Sajadmanesh and Gatica{-}Perez(2021)]%
        {LPGNN}
\bibfield{author}{\bibinfo{person}{Sina Sajadmanesh} {and}
  \bibinfo{person}{Daniel Gatica{-}Perez}.} \bibinfo{year}{2021}\natexlab{}.
\newblock \showarticletitle{Locally Private Graph Neural Networks}. In
  \bibinfo{booktitle}{\emph{CCS}}, \bibfield{editor}{\bibinfo{person}{Yongdae
  Kim}, \bibinfo{person}{Jong Kim}, \bibinfo{person}{Giovanni Vigna}, {and}
  \bibinfo{person}{Elaine Shi}} (Eds.). \bibinfo{publisher}{{ACM}},
  \bibinfo{pages}{2130--2145}.
\newblock
\urldef\tempurl%
\url{https://doi.org/10.1145/3460120.3484565}
\showDOI{\tempurl}


\bibitem[Sankar et~al\mbox{.}(2020)]%
        {exp_dataset_weeplaces}
\bibfield{author}{\bibinfo{person}{Aravind Sankar}, \bibinfo{person}{Yanhong
  Wu}, \bibinfo{person}{Yuhang Wu}, \bibinfo{person}{Wei Zhang},
  \bibinfo{person}{Hao Yang}, {and} \bibinfo{person}{Hari Sundaram}.}
  \bibinfo{year}{2020}\natexlab{}.
\newblock \showarticletitle{GroupIM: {A} Mutual Information Maximization
  Framework for Neural Group Recommendation}. In
  \bibinfo{booktitle}{\emph{SIGIR}}. \bibinfo{publisher}{{ACM}},
  \bibinfo{pages}{1279--1288}.
\newblock
\urldef\tempurl%
\url{https://doi.org/10.1145/3397271.3401116}
\showDOI{\tempurl}


\bibitem[Shamir(1979)]%
        {secret_sharing}
\bibfield{author}{\bibinfo{person}{Adi Shamir}.}
  \bibinfo{year}{1979}\natexlab{}.
\newblock \showarticletitle{How to Share a Secret}.
\newblock \bibinfo{journal}{\emph{Commun. {ACM}}} \bibinfo{volume}{22},
  \bibinfo{number}{11} (\bibinfo{year}{1979}), \bibinfo{pages}{612--613}.
\newblock
\urldef\tempurl%
\url{https://doi.org/10.1145/359168.359176}
\showDOI{\tempurl}


\bibitem[Shen et~al\mbox{.}(2018)]%
        {soft_attention}
\bibfield{author}{\bibinfo{person}{Tao Shen}, \bibinfo{person}{Tianyi Zhou},
  \bibinfo{person}{Guodong Long}, \bibinfo{person}{Jing Jiang},
  \bibinfo{person}{Sen Wang}, {and} \bibinfo{person}{Chengqi Zhang}.}
  \bibinfo{year}{2018}\natexlab{}.
\newblock \showarticletitle{Reinforced Self-Attention Network: a Hybrid of Hard
  and Soft Attention for Sequence Modeling}. In
  \bibinfo{booktitle}{\emph{IJCAI}},
  \bibfield{editor}{\bibinfo{person}{J{\'{e}}r{\^{o}}me Lang}} (Ed.).
  \bibinfo{publisher}{ijcai.org}, \bibinfo{pages}{4345--4352}.
\newblock
\urldef\tempurl%
\url{https://doi.org/10.24963/ijcai.2018/604}
\showDOI{\tempurl}


\bibitem[Su et~al\mbox{.}(2021)]%
        {GCMC}
\bibfield{author}{\bibinfo{person}{Chang Su}, \bibinfo{person}{Min Chen}, {and}
  \bibinfo{person}{Xianzhong Xie}.} \bibinfo{year}{2021}\natexlab{}.
\newblock \showarticletitle{Graph Convolutional Matrix Completion via Relation
  Reconstruction}. In \bibinfo{booktitle}{\emph{ICSCA}}.
  \bibinfo{publisher}{{ACM}}, \bibinfo{pages}{51--56}.
\newblock
\urldef\tempurl%
\url{https://doi.org/10.1145/3457784.3457792}
\showDOI{\tempurl}


\bibitem[Vaswani et~al\mbox{.}(2017)]%
        {attention}
\bibfield{author}{\bibinfo{person}{Ashish Vaswani}, \bibinfo{person}{Noam
  Shazeer}, \bibinfo{person}{Niki Parmar}, \bibinfo{person}{Jakob Uszkoreit},
  \bibinfo{person}{Llion Jones}, \bibinfo{person}{Aidan~N. Gomez},
  \bibinfo{person}{Lukasz Kaiser}, {and} \bibinfo{person}{Illia Polosukhin}.}
  \bibinfo{year}{2017}\natexlab{}.
\newblock \showarticletitle{Attention is All you Need}. In
  \bibinfo{booktitle}{\emph{NIPS}}. \bibinfo{pages}{5998--6008}.
\newblock


\bibitem[Wang et~al\mbox{.}(2019)]%
        {NGCF}
\bibfield{author}{\bibinfo{person}{Xiang Wang}, \bibinfo{person}{Xiangnan He},
  \bibinfo{person}{Meng Wang}, \bibinfo{person}{Fuli Feng}, {and}
  \bibinfo{person}{Tat{-}Seng Chua}.} \bibinfo{year}{2019}\natexlab{}.
\newblock \showarticletitle{Neural Graph Collaborative Filtering}. In
  \bibinfo{booktitle}{\emph{SIGIR}},
  \bibfield{editor}{\bibinfo{person}{Benjamin Piwowarski}, \bibinfo{person}{Max
  Chevalier}, \bibinfo{person}{{\'{E}}ric Gaussier}, \bibinfo{person}{Yoelle
  Maarek}, \bibinfo{person}{Jian{-}Yun Nie}, {and} \bibinfo{person}{Falk
  Scholer}} (Eds.). \bibinfo{publisher}{{ACM}}, \bibinfo{pages}{165--174}.
\newblock
\urldef\tempurl%
\url{https://doi.org/10.1145/3331184.3331267}
\showDOI{\tempurl}


\bibitem[Wang et~al\mbox{.}(2020)]%
        {DGCF}
\bibfield{author}{\bibinfo{person}{Xiang Wang}, \bibinfo{person}{Hongye Jin},
  \bibinfo{person}{An Zhang}, \bibinfo{person}{Xiangnan He},
  \bibinfo{person}{Tong Xu}, {and} \bibinfo{person}{Tat{-}Seng Chua}.}
  \bibinfo{year}{2020}\natexlab{}.
\newblock \showarticletitle{Disentangled Graph Collaborative Filtering}. In
  \bibinfo{booktitle}{\emph{SIGIR}}. \bibinfo{publisher}{{ACM}},
  \bibinfo{pages}{1001--1010}.
\newblock
\urldef\tempurl%
\url{https://doi.org/10.1145/3397271.3401137}
\showDOI{\tempurl}


\bibitem[Weisstein(2006)]%
        {correlation}
\bibfield{author}{\bibinfo{person}{Eric~W Weisstein}.}
  \bibinfo{year}{2006}\natexlab{}.
\newblock \showarticletitle{Correlation coefficient}.
\newblock \bibinfo{journal}{\emph{MathWorld}} (\bibinfo{year}{2006}).
\newblock


\bibitem[Wu et~al\mbox{.}(2021b)]%
        {FedGNN}
\bibfield{author}{\bibinfo{person}{Chuhan Wu}, \bibinfo{person}{Fangzhao Wu},
  \bibinfo{person}{Yang Cao}, \bibinfo{person}{Yongfeng Huang}, {and}
  \bibinfo{person}{Xing Xie}.} \bibinfo{year}{2021}\natexlab{b}.
\newblock \showarticletitle{FedGNN: Federated Graph Neural Network for
  Privacy-Preserving Recommendation}.
\newblock \bibinfo{journal}{\emph{CoRR}}  \bibinfo{volume}{abs/2102.04925}
  (\bibinfo{year}{2021}).
\newblock
\showeprint[arXiv]{2102.04925}
\urldef\tempurl%
\url{https://arxiv.org/abs/2102.04925}
\showURL{%
\tempurl}


\bibitem[Wu et~al\mbox{.}(2021a)]%
        {SGL}
\bibfield{author}{\bibinfo{person}{Jiancan Wu}, \bibinfo{person}{Xiang Wang},
  \bibinfo{person}{Fuli Feng}, \bibinfo{person}{Xiangnan He},
  \bibinfo{person}{Liang Chen}, \bibinfo{person}{Jianxun Lian}, {and}
  \bibinfo{person}{Xing Xie}.} \bibinfo{year}{2021}\natexlab{a}.
\newblock \showarticletitle{Self-supervised graph learning for recommendation}.
  In \bibinfo{booktitle}{\emph{SIGIR}}. \bibinfo{pages}{726--735}.
\newblock


\bibitem[Wu et~al\mbox{.}(2019)]%
        {Sequence_GNN}
\bibfield{author}{\bibinfo{person}{Shu Wu}, \bibinfo{person}{Yuyuan Tang},
  \bibinfo{person}{Yanqiao Zhu}, \bibinfo{person}{Liang Wang},
  \bibinfo{person}{Xing Xie}, {and} \bibinfo{person}{Tieniu Tan}.}
  \bibinfo{year}{2019}\natexlab{}.
\newblock \showarticletitle{Session-Based Recommendation with Graph Neural
  Networks}. In \bibinfo{booktitle}{\emph{AAAI}}. \bibinfo{pages}{346--353}.
\newblock
\urldef\tempurl%
\url{https://doi.org/10.1609/aaai.v33i01.3301346}
\showDOI{\tempurl}


\bibitem[Yang et~al\mbox{.}(2019)]%
        {Prisvr}
\bibfield{author}{\bibinfo{person}{Dingqi Yang}, \bibinfo{person}{Bingqing Qu},
  {and} \bibinfo{person}{Philippe Cudr{\'{e}}{-}Mauroux}.}
  \bibinfo{year}{2019}\natexlab{}.
\newblock \showarticletitle{Privacy-Preserving Social Media Data Publishing for
  Personalized Ranking-Based Recommendation}.
\newblock \bibinfo{journal}{\emph{{IEEE} Trans. Knowl. Data Eng.}}
  \bibinfo{volume}{31}, \bibinfo{number}{3} (\bibinfo{year}{2019}),
  \bibinfo{pages}{507--520}.
\newblock
\urldef\tempurl%
\url{https://doi.org/10.1109/TKDE.2018.2840974}
\showDOI{\tempurl}


\bibitem[Ying et~al\mbox{.}(2018a)]%
        {pinsage}
\bibfield{author}{\bibinfo{person}{Rex Ying}, \bibinfo{person}{Ruining He},
  \bibinfo{person}{Kaifeng Chen}, \bibinfo{person}{Pong Eksombatchai},
  \bibinfo{person}{William~L Hamilton}, {and} \bibinfo{person}{Jure Leskovec}.}
  \bibinfo{year}{2018}\natexlab{a}.
\newblock \showarticletitle{Graph convolutional neural networks for web-scale
  recommender systems}. In \bibinfo{booktitle}{\emph{SIGKDD}}.
  \bibinfo{pages}{974--983}.
\newblock


\bibitem[Ying et~al\mbox{.}(2018b)]%
        {Differential_Pool}
\bibfield{author}{\bibinfo{person}{Zhitao Ying}, \bibinfo{person}{Jiaxuan You},
  \bibinfo{person}{Christopher Morris}, \bibinfo{person}{Xiang Ren},
  \bibinfo{person}{William~L. Hamilton}, {and} \bibinfo{person}{Jure
  Leskovec}.} \bibinfo{year}{2018}\natexlab{b}.
\newblock \showarticletitle{Hierarchical Graph Representation Learning with
  Differentiable Pooling}. In \bibinfo{booktitle}{\emph{NIPS}}.
  \bibinfo{pages}{4805--4815}.
\newblock


\bibitem[Zhou et~al\mbox{.}(2006)]%
        {first_introduce_hypergraph}
\bibfield{author}{\bibinfo{person}{Dengyong Zhou}, \bibinfo{person}{Jiayuan
  Huang}, {and} \bibinfo{person}{Bernhard Sch{\"o}lkopf}.}
  \bibinfo{year}{2006}\natexlab{}.
\newblock \showarticletitle{Learning with hypergraphs: Clustering,
  classification, and embedding}.
\newblock \bibinfo{journal}{\emph{NIPS}}  \bibinfo{volume}{19}
  (\bibinfo{year}{2006}).
\newblock


\bibitem[Zhou et~al\mbox{.}(2022)]%
        {MPC2}
\bibfield{author}{\bibinfo{person}{Xing Zhou}, \bibinfo{person}{Zhilei Xu},
  \bibinfo{person}{Cong Wang}, {and} \bibinfo{person}{Mingyu Gao}.}
  \bibinfo{year}{2022}\natexlab{}.
\newblock \showarticletitle{PPMLAC: high performance chipset architecture for
  secure multi-party computation}. In \bibinfo{booktitle}{\emph{ISCA}}.
  \bibinfo{pages}{87--101}.
\newblock


\bibitem[Ziegler et~al\mbox{.}(2005)]%
        {exp_dataset_book_crossing}
\bibfield{author}{\bibinfo{person}{Cai{-}Nicolas Ziegler},
  \bibinfo{person}{Sean~M. McNee}, \bibinfo{person}{Joseph~A. Konstan}, {and}
  \bibinfo{person}{Georg Lausen}.} \bibinfo{year}{2005}\natexlab{}.
\newblock \showarticletitle{Improving recommendation lists through topic
  diversification}. In \bibinfo{booktitle}{\emph{WWW}}.
  \bibinfo{pages}{22--32}.
\newblock
\urldef\tempurl%
\url{https://doi.org/10.1145/1060745.1060754}
\showDOI{\tempurl}


\end{thebibliography}


\end{document}